\RequirePackage{amsmath}
\documentclass[runningheads]{llncs}

\usepackage[T1]{fontenc}
\usepackage{graphicx}
\usepackage{amssymb}
\usepackage{reptheorem}
\usepackage{cleveref}
\usepackage{todonotes}
\usepackage{tikzsymbols}
\usepackage{tikz}
\usetikzlibrary{arrows}
\usetikzlibrary{calc}
\usetikzlibrary{patterns}
\usetikzlibrary{positioning}
\usetikzlibrary{decorations.pathreplacing}
\usepackage{algorithm}
\usepackage{algpseudocode}


\crefname{cclaim}{Claim}{Claims}
\spnewtheorem{cclaim}[claims]{Claim}{\itshape}{\rmfamily}
\spnewtheorem*{proofsketch}{Proof (sketch)}{\itshape}{\rmfamily}

\newcommand{\suchthat}{\;\ifnum\currentgrouptype=16 \middle\fi|\;}

\begin{document}

\title{Arithmetic Circuits with Division} 
\author{Silas Cato Sacher\orcidID{0009-0004-6850-1298}}
\authorrunning{S. C. Sacher}
\institute{Universit\"at Trier, Fachbereich IV, Informatikwissenschaften, Germany \\ 
\email{sacher@uni-trier.de}}

\maketitle

\newcommand{\logReduction}[0]{\leq^{\log}_m}
\newcommand{\logEquiv}[0]{\equiv^{\log}_m}
\newcommand{\pReduction}[0]{\leq^{\mathrm{P}}_m}
\newcommand{\pspaceReduction}[0]{\leq^{\mathrm{PSPACE}}_m}
\newcommand{\natNum}[0]{\mathbb{N}}

\newcommand{\DTIME}[0]{\mathrm{DTIME}}
\newcommand{\NTIME}[0]{\mathrm{NTIME}}
\newcommand{\DSPACE}[0]{\mathrm{DSPACE}}
\newcommand{\NSPACE}[0]{\mathrm{NSPACE}}
\newcommand{\ATIME}[0]{\mathrm{ATIME}}
\newcommand{\E}[0]{\mathrm{E}}
\newcommand{\EXP}[0]{\mathrm{EXP}}
\newcommand{\NP}[0]{\mathrm{NP}}
\newcommand{\NEXP}[0]{\mathrm{NEXP}}
\newcommand{\PSPACE}[0]{\mathrm{PSPACE}}
\newcommand{\PL}[0]{\mathrm{PL}}
\newcommand{\NL}[0]{\mathrm{NL}}
\newcommand{\PIT}[0]{\mathcal{PIT}}
\newcommand{\CEqualsL}[0]{\mathrm{C_{=}L}}
\newcommand{\defEquals}[0]{:=}

\newcommand{\compl}[0]{\overline{\phantom{c}}}
\newcommand{\cO}[0]{\mathcal{O}}
\newcommand{\MC}[0]{\mathrm{MC}}
\newcommand{\primes}[0]{\mathbb{P}}
\newcommand{\natNumInf}[0]{\natNum_{\infty}}
\newcommand{\ExactCover}[0]{\textsc{ExactCover}}

\newcommand{\ifAndOnlyIf}[0]{if and only if }
\newcommand{\logSpace}[0]{logarithmic space }
\newcommand{\polyTime}[0]{polynomial time }
\newcommand{\polySpace}[0]{polynomial space }
\newcommand{\TM}[0]{Turing machine } 
\newcommand{\TMs}[0]{Turing machines } 

\newcommand{\myNote}[1]{\todo[color=gray!20]{#1}}

\begin{abstract}
We study the computational complexity of the membership problem for arithmetic circuits over natural numbers with division. We consider different subsets of the operations $\{\cup,\cap, \compl, +, \times, / \}$, where $/$ is the element-wise integer division (without remainder and without rounding). Results for the subsets without division have been studied before, in particular by McKenzie and Wagner~\cite{MW07} and Yang~\cite{Yang2000}. The division is expressive because it makes it possible to describe the set of factors of a given number as a circuit. Surprisingly, the cases $\{ \cup,\cap,\compl,+, / \}$ and $\{ \cup,\cap,\compl,\times, / \}$ are $\PSPACE$-complete and therefore equivalent to the corresponding cases without division. The case $\{ \cup, / \}$ is $\NP$-hard in contrast to the case $\{ \cup \}$ which is $\NL$-complete. Further upper bounds, lower bounds and completeness results are given.
\end{abstract} 
\section{Introduction} 
Arithmetic circuits are a generalization of formulas over sets of natural numbers (including $0$). In such formulas, the numbers are interpreted as singleton sets that are combined through some operations (initially only $\{\cap, \cup, \compl, + \}$) similar to the way regular languages can be expressed by regular expressions. For example $2 \times \overline{0 \cap 1}$ describes the set of even numbers, because $\overline{0 \cap 1}$ describes the set of natural numbers. In fact, such formulas were introduced by Stockmeyer and Meyer in 1973 in a work that otherwise considered regular expressions~\cite{SM73}. The arithmetic circuits was first considered by Wagner in 1984~\cite{Wa84}. 

An arithmetic circuit $C= (V, E, g_c, \alpha)$ is an acyclic directed graph, with a map $\alpha : V \rightarrow \cO \cup \natNum$ labeling the nodes (called gates) and a fixed output gate $g_c$. To the gates, the map assigns either natural numbers (interpreted as singleton sets) or an operation that describes how the sets from the predecessor nodes are combined into a new set. Arithmetic operations are applied element-wise, for example, $A + B = \{ c \in \natNum \mid \exists a \in A \; \exists b \in B: c = a + b \}$ for sets $A, B \subseteq \natNum$. Hence circuits are essentially formulas in which parts of the formula can be reused. Therefore, circuits can provide a more compact representation in certain cases. Wagner proposed the question of computational complexity of the membership problem for circuits with the operations $\{ \cup, +, \times \}$. The membership problem on the instance $(C,b)$ is the question whether the natural number $b$ is contained in the set described by the circuit $C$. In 2000 Yang proved that the problem described by Wagner is $\PSPACE$-complete~\cite{Yang2000}. 
Based on this, McKenzie and Wagner in 2003 provided a detailed description of the complexity of the membership problem on natural numbers for the various subsets of the operations $\{\cap, \cup, \compl, +, \times\}$~\cite{MW07}. However, some question--including the decidability of the general problem--remained open. They also provided insight into why determining the complexity of the general problem is difficult: It is possible to describe the set of prime numbers $\primes$ with the arithmetic circuit shown in \Cref{fig:primes} ~\cite[Example 2.1]{MW07}. 
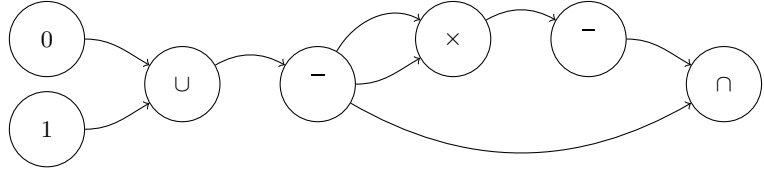
\begin{figure}
\begin{center}
\begin{tikzpicture}
	\newcommand\dx{1.8}
	\newcommand\dy{0.6}
	\newcommand\minSize{1cm}
	
	\node[draw,circle,fill=white,radius=0.5, minimum size=\minSize] (inA) at (0*\dx, -0*\dy) {$0$};
	
	\node[draw,circle,fill=white,radius=0.5, minimum size=\minSize] (inB) at (0*\dx, -2*\dy) {$1$};
	
	\node[draw,circle,fill=white,radius=0.5, minimum size=\minSize] (cup) at (1*\dx, -1*\dy) {$\cup$};
	
	\node[draw,circle,fill=white,radius=0.5, minimum size=\minSize] (negA) at (2*\dx, -1*\dy) {$\overline{\phantom{0}}$};
	
	\node[draw,circle,fill=white,radius=0.5, minimum size=\minSize] (times) at (3*\dx, -0*\dy) {$\times$};
	
	\node[draw,circle,fill=white,radius=0.5, minimum size=\minSize] (negB) at (4*\dx, -0*\dy) {$\overline{\phantom{0}}$};
	
	\node[draw,circle,fill=white,radius=0.5, minimum size=\minSize] (cap) at (5*\dx, -1*\dy) {$\cap$};
	
	\draw[->] (inA)[out=0,in=150] to (cup);
	\draw[->] (inB)[out=0,in=-150] to (cup);
	\draw[->] (cup)[out=30,in=150] to (negA);
	\draw[->] (negA)[out=60,in=150] to (times);
	\draw[->] (negA)[out=0,in=-150] to (times);
	\draw[->] (times)[out=30,in=150] to (negB);
	\draw[->] (negB)[out=0,in=150] to (cap);
	\draw[->] (negA)[out=-30,in=-150] to (cap);
\end{tikzpicture} 
\end{center}
\caption{Arithmetic circuit for the set of prime numbers} \label{fig:primes}
\end{figure}
Consider Goldbach's conjecture. This conjecture states that every even natural number greater than two is the sum of two prime numbers. It was first proposed in a letter exchange between Goldbach and Euler in 1742 and is a part of Hilbert's eighth problem. Glaßer first observed that one can construct a circuit that produces the number $0$ in its output gate, \ifAndOnlyIf Goldbach's conjecture holds~\cite{MW07}. Therefore an algorithm for the general membership problem could determine whether Goldbach's conjecture holds or not. 

There has been extensive research on arithmetic circuits: Arithmetic circuits with different base sets (such as integers~\cite{Tra06}, positive numbers~\cite{Br07}, balanced sets (specific subsets of $\natNum$)~\cite{Dose2020} or even functions~\cite{PHD09}) were studied. Different problems such as the equivalence problem~\cite{GHRTW08} (``Do two given circuits describe the same set?'') and the emptiness problem~\cite{TR17:2} (``Does a given circuit describe the empty set?'') have been studied as well. There have also been generalizations of the membership problems: In 2007 Glaßer et al. considered satisfiability problems for circuits over sets of natural numbers~\cite{Gla07}. Furthermore, arithmetic circuits over semi-rings in non-commutative settings have been considered as well \cite{Mahajan94}. 

In this work, we study the computational complexity of the membership problem for circuits over natural numbers with division. We will consider the different sets of operations $\cO$ such that $\{ / \} \subseteq \cO \subseteq\{\cap, \cup, \compl, +, \times, / \}$, where $/$ is the element-wise division without remainder or rounding, that is,  for sets $A, B \subseteq \natNum$, $ A / B = \{ c \in \natNum \mid \exists a \in A \; \exists b \in B \setminus \{ 0 \}: a = c \cdot b  \}$. 
The author found this operation to be especially expressive compared to division with rounding, because one can easily express sets like the set of factors of a given number $b$ (Consider the formula $b / \overline{0 \cap 1}$.), the prime factors of a given number and the common divisors of two given numbers. We were already able to express the set of all multiples of a given number and the set of numbers that are the product of exactly $k$ prime numbers with multiplication. These properties bring the membership problem of arithmetic circuits with division (but without addition) close to Skolem arithmetic. 
Skolem arithmetic is decidable, but the best known algorithm has a triple exponential running time. 
One can reduce the membership problem for the operations $\{\cap, \cup, \compl, \times, / \}$ to Skolem arithmetic. This can be proven in a way similar to the results of Glaßer et al.~\cite[Thm. 1]{Gl15}. However, the proof will not be given in this work because instead we will prove that the problem is $\PSPACE$-complete. Indeed we will show that the cases $\{ \cup,\cap,\compl,+, / \}$ and $\{ \cup,\cap,\compl,\times, / \}$ are $\PSPACE$-complete and equivalent to the corresponding cases without division. In addition, we prove that the case $\{ \cup,\cap,+, \times, / \}$ is $\NEXP$-complete and equivalent to the case $\{ \cup,\cap,+, \times \}$. The case $\{ \cup, / \}$ is $\NP$-hard in contrast to the case $\{ \cup \}$ which is $\NL$-complete. Therefore, adding division will--under common assumptions--increase the computational complexity. Further upper bounds will be given in \Cref{Kapitel1} and further lower bounds will be given in \Cref{Kapitel2}. 

A detailed overview of the results (including references to the respective theorems) is given in \Cref{Tabelle}. If the lower bound in the line for the set of operations $\cO$ has an external reference, then it is a reference for a hardness result for the membership problem for circuits with operations $\cO'$ such that $\mathcal{O}' \subseteq \cO \setminus \{ / \}$. This can be transferred to the membership problem with the operations $\cO$. 

The results of this work are also contained in the author's master's thesis~\cite{Sacher2023}.
\begin{table}[t]
\centering
\caption{Lower and upper bounds of membership problems with division.} \label{Tabelle}
\begin{tabular}[h]{l|llllll|ll|ll}
& $\cO$ & & & & & & Lower bound & Reference & Upper bound & Reference \\
\hline
1 & $\cup$ & $\cap$ & $\compl$ & $+$ & $\times$ & $/$ & $\NEXP$ & \cite[Th. 6.4]{MW07} & ? & \\
2 &        &        & $\compl$ & $+$ & $\times$ & $/$ & $\PSPACE$ & \cite[Th. 6 (2)]{Tra06} & ?  & \\
3 & $\cup$ & $\cap$ &          & $+$ & $\times$ & $/$ & $\NEXP$ & \cite[Th. 6.2]{MW07} & $\NEXP$ & \Cref{Satz6}\\
4 & $\cup$ &        &          & $+$ & $\times$ & $/$ & $\PSPACE$ & \cite[Th. 2.3]{MW07} & $\NEXP$ & \Cref{Satz6}\\
5 &        & $\cap$ &          & $+$ & $\times$ & $/$ & $\PIT$ & \Cref{Folg3} & $\EXP$ & \Cref{Satz3} \\
6 &        &        &          & $+$ & $\times$ & $/$ & $\PIT$ & \Cref{Folg7} & $\EXP$ & \Cref{Satz3} \\
\hline
 7 & $\cup$ & $\cap$ & $\compl$ & $+$ & & $/$ & $\PSPACE$ & \cite[Th. 5.5]{MW07} & $\PSPACE$ & \Cref{Satz1}\\
 8 &        &        & $\compl$ & $+$ & & $/$ & $\PSPACE$ & \cite[Th. 6 (2)]{Tra06} & $\PSPACE$ & \Cref{Satz1}\\
 9 & $\cup$ & $\cap$ &          & $+$ & & $/$ & $\PSPACE$ & \cite[Th. 5.5]{MW07} & $\PSPACE$ & \Cref{Satz1}\\
10 & $\cup$ &        &          & $+$ & & $/$ & $\NP$ & \cite[Th. 4.4]{MW07} & $\PSPACE$ & \Cref{Satz1}\\
11 &        & $\cap$ &          & $+$ & & $/$ & $\CEqualsL$ & \cite[Th. 8.2]{MW07} & $\mathrm{P}$ & \Cref{Satz3}\\
12 &        &        &          & $+$ & & $/$ & $\CEqualsL$& \cite[Th. 8.2]{MW07} & $\mathrm{P}$ & \Cref{Satz3}\\
\hline
13 & $\cup$ & $\cap$ & $\compl$ & & $\times$ & $/$ & $\PSPACE$ & \cite[Th. 5.5]{MW07} & $\PSPACE$ & \Cref{Folg5}\\
14 &        &        & $\compl$ & & $\times$ & $/$ & $\PSPACE$ & \cite[Lem. 26]{TR17:2} & $\PSPACE$ & \Cref{Folg5}\\
15 & $\cup$ & $\cap$ &          & & $\times$ & $/$ & $\PSPACE$ & \cite[Th. 5.5]{MW07} & $\PSPACE$ & \Cref{Folg5}\\
16 & $\cup$ &        &          & & $\times$ & $/$ & $\NP$ & \cite[Th. 4.4]{MW07} & $\PSPACE$ & \Cref{Folg5}\\
17 &        & $\cap$ &          & & $\times$ & $/$ & $\PL$ & \Cref{Satz4} & $\mathrm{P}$ & \Cref{Folg4}\\
18 &        &        &          & & $\times$ & $/$ & $\PL$ & \Cref{Satz4} & $\mathrm{P}$ & \Cref{Folg4}\\
\hline
19 & $\cup$ & $\cap$ & $\compl$ & & & $/$ & $\NP$ & \Cref{Satz5} & $\PSPACE$ & \Cref{Satz1}\\
20 &        &        & $\compl$ & & & $/$ & $\mathrm{P}$ & \Cref{Satz8} & $\PSPACE$ & \Cref{Satz1}\\
21 & $\cup$ & $\cap$ &          & & & $/$ & $\NP$ & \Cref{Satz5} & $\PSPACE$ & \Cref{Satz1}\\
22 & $\cup$ &        &          & & & $/$ & $\NP$ & \Cref{Satz5} & $\PSPACE$ & \Cref{Satz1}\\
23 &        & $\cap$ &          & & & $/$ & $\NL$ & \cite[Th. 9.1]{MW07} & $\mathrm{P}$ & \Cref{Satz3}\\
24 &        &        &          & & & $/$ & $\NL$ & \Cref{Satz2} & $\mathrm{P}$ & \Cref{Satz3}\\
\end{tabular}
\label{results}
\end{table}
\section{Preliminaries} \label{Grundlagen}
We consider $0$ to be a natural number and define the function $\log: \natNum \rightarrow \natNum$ as $\log(n) = \lfloor \log_2(n) \rfloor$ for all $n > 0$ and $\log( 0) = 0$. We write the set of prime numbers as $\primes$. For $k \in \natNum$, $[k] \defEquals \{ 1, \dots, k \}$. 
We assume that the reader is familiar with several complexity classes, especially the deterministic time complexity classes $\mathrm{P}$, $\E = \DTIME( 2^{O(n)})$ and $\EXP = \bigcup_{k \geq 1}\DTIME( 2^{O(n^k)})$, the nondeterministic time complexity classes $\NP$ and $\NEXP = \bigcup_{k \geq 1} \NTIME(2^{O(n^k)})$, the space complexity classes $\mathrm{L}$, $\NL$, $\PSPACE$ and the probabilistic complexity class $\PL$. 
$\ATIME(t)$ is the set of all languages accepted by an alternating Turing machine--a model for parallel computers--in time $t: \natNum \rightarrow \natNum$. An introduction to alternating \TMs in the context of circuits can be found in \cite[Section 2.5]{Vollmer99}. 
$\mathrm{FL}$ is the class of all functions $f: \Sigma_1^* \rightarrow \Sigma_2^*$ for alphabets $\Sigma_1$ and $\Sigma_2$ calculated in space $O(\log(n))$ by a deterministic (multi-tape) \TM with output. We write $A \logReduction B$, \ifAndOnlyIf the language $A$ is many-one-reducible to a language $B$ in \logSpace and $A \logEquiv B$, \ifAndOnlyIf $A \logReduction B$ and $B \logReduction A$. It is known that $\mathrm{FL}$ is stable under composition and that $\logReduction$ is reflexive and transitive. In addition to $\logReduction$, we sometimes consider $\pReduction$ for \polyTime and $\pspaceReduction$ for \polySpace reduction. We investigate lower and upper bounds for membership problems. By a lower bound $\mathcal{C}_1$ we mean that a problem $A$ is $\leq^{log}_m$-hard for the complexity class $\mathcal{C}_1$ and by an upper bound $\mathcal{C}_2$ we mean that $A \in \mathcal{C}_2$. In this work, we say that a problem $A$ is $\mathcal{C}$-hard (or $\mathcal{C}$-complete) if $ A$ is $\logReduction$-hard (or $\logReduction$-complete) for $\mathcal{C}$. 

\subsection{Arithmetic Circuits}
\begin{definition} 
Let $\emptyset \neq \cO \subseteq \{ \cup, \cap , \compl, +, \times, / \}$ be a set of operations on sets of natural numbers. The operations $\cup$, $\cap$ and $\compl$ are the union, intersection and complement for sets of natural numbers. We write the multiplication as $\times$. The operations $\sigma \in \{ +, \times, / \}$ are extended for sets $A, B \subseteq \natNum$:
$ A \sigma B \defEquals \{ c \in \mathbb {N} \; | \; \exists a \in A, b \in B: \; c = a \sigma b \}$. 
An \emph{(arithmetic) $\cO$-circuit} $C= (V, E, g_C, \alpha)$ is an acyclic, directed multigraph $G = (V,E)$, whose nodes have indegree 0, 1 or 2. Let $V \subseteq \natNum$. The nodes of $G$ are called \emph{gates}. Gates with indegree 0 are called \emph{input gates} and the specified gate $g_C \in V$ is called \emph{output gate}. The gates are labeled with the function $\alpha : V \rightarrow \cO \cup \natNum$, such that for each input gate $g$ we have $\alpha(g) \in \natNum$, for each gate $g$ with indegree $1$ we have $\alpha(g) = \compl$ (only if $\compl \in \cO$) and for all other gates we have $\alpha(g) \in  \cO \setminus \{ \compl \}$. A gate $g$ is also called an \emph{$\alpha(g)$-gate}. For each gate $g \in V$ we define the \emph{result set} $I(g) \subseteq \natNum$. For each input gate $g$, $I(g) \defEquals \{ \alpha(g) \}$. For a gate $g$ with label $\alpha(g) = \compl$ and predecessor $p$, $I(g) \defEquals \mathbb {N} \setminus I(p)$. We define the predecessor function $p(g) \defEquals p$ for $\compl$-gates. For a gate $g$ with label $\alpha(g) = \sigma \in \{ \cup, \cap , +, \times, / \}$ and predecessors $g_1$ and $g_2$ with $g_1 \leq g_2$, $I(g) \defEquals I(g_1) \sigma I(g_2)$. We call $g_1$ the first and $g_2$ the second predecessor of $g$ and define the predecessor functions $p_1(g) \defEquals g_1$ and $p_2(g) \defEquals g_2$. $I(C) \defEquals I(g_c)$ is the \emph{result set of $C$}. 
\end{definition} 

\begin{definition} 
For $\emptyset \neq \cO \subseteq \{ \cup, \cap , \compl, +, \times, / \}$, the \emph{membership problem for $\cO$-circuits} is 
$ \MC(\mathcal {O}) \defEquals \{ (C,b) \; | \; C \text{ is an } \mathcal {O} \text{-circuit and } b \in \mathbb {N} \text{ such that } b \in I(C) \} $.
We use the shorthand notation $\MC(\sigma_1, \dots, \sigma_n)$ for $\MC(\{ \sigma_1, \dots, \sigma_n \})$. 
\end{definition}

\begin{remark}
For $\emptyset \neq \cO_1 \subseteq \cO_2 \subseteq \{ \cup, \cap , \compl, +, \times, / \}$, $\MC(\cO_1) \logReduction \MC(\cO_2)$.
\end{remark}

\subsubsection{Encoding of circuits} In order to determine the complexity of the membership problem, we must understand circuits as inputs of Turing machines. To do this, we encode a circuit $C= (V, E, g_c, \alpha)$ as follows: The circuit is represented as a list of its gates in reverse topological order. Each gate $g \in V$ has the representation $( g, \alpha(g), l(g))$. Where $l(g)$ is the list of predecessors of $g$. If $g$ has an ingoing double edge, the corresponding predecessor is enumerated twice. 
We encode this list in bits. (Hence, the integer labels of the input gates are encoded in binary.) We write the length of the encoding as $|C|$. The value $|C|$ depends on both the size of the underlying multigraph and the labels of the input gates. Topological ordering of the edges is possible because the graph is acyclic. It is possible to deterministically check in \logSpace whether the order is topological. This makes it possible to deterministically check in \logSpace whether the input is a valid encoding of a circuit.
\section{Upper Bounds} \label{Kapitel1}
We do not solve the general problem $\MC( \cup, \cap , \compl, +, \times, / )$. The open problem $\MC( \cup, \cap , \compl, +, \times)$ is reducible to it using logarithmic space. Instead we look at problems without addition, without complement and without multiplication, respectively. 

\subsection{Circuits without multiplication}
We will solve the membership problem for circuits without multiplication on alternating Turing machines. First, we need a technical lemma: Glaßer et al. state that for $\{ \cup, \cap , \compl, + \}$-circuits there is a limit beyond which either all values are described by the circuit or none \cite[Prop. 1]{GHRTW08}. This is preserved if we add $/$-gates. We prove the following by induction:
\begin{makethm}{lemma}{Lemma1}
Let $C$ be a $\cO$-circuit with $\emptyset \not = \cO \subseteq \{ \cup, \cap , \compl, +, / \}$. A number $n \leq 2^{|C|} + 1$ exists such that for all $z \geq n$ it holds that $ z \in I(C)$ \ifAndOnlyIf $ n \in I(C)$.
\end{makethm}

\begin{proof} 
Let $C$ be a $\cO$-circuit. Let $C_g$ be the circuit formed from $C$ using $g$ as the output gate and removing all gates from which $g$ cannot be reached. We show for each gate $g$ from $C$ that a natural number $n_g \leq 2^{|C_g|} +1$ exists such that for each $z \geq n_g$, $ z \in I(g) $ \ifAndOnlyIf $ n_g \in I(g) $. We show this statement by induction. \\
\underline{Base case:} Consider an input gate $g$. The gate $g$ is the only gate in $C_g$ and $I(g) = \{ \alpha(g) \}$. Define $n_g = 2^{|C_g|} +1$. For each $z \geq n_g$, $z \geq 2^{\log(\alpha(g))+1} +1 > \alpha(g)$ and $z \notin I(g)$. \\
\underline{Induction step:} Let $g$ be a gate with predecessors. We assume that the statement holds for the predecessors of $g$ (induction hypothesis). We only consider the case where $g$ is a $/$-gate, because Glaßer et al. already proved the statement of the lemma for $\{ \cup, \cap , \compl, + \}$-circuits \cite[Prop. 1]{GHRTW08}. Let $n_g \defEquals n_{p_1(g)}$. It holds that
$n_g = n_{p_1(g)} \leq 2^{|C_{p_1(g)}|} +1 \leq 2^{|C_g|} +1 $. First, consider the case $n_{p_1(g)} \notin I(p_1(g))$. Hence $I(p_1(g)) \subseteq [0, n_{p_1(g)})$ and $I(g) = I(p_1(g)) / I(p_2(g)) \subseteq [0, n_g) = [0, n_{p_1(g)}) $. Now consider the case $n_{p_1(g)} \in I(p_1(g))$. If $I(p_2(g)) \subseteq \{ 0 \}$, $I(g)$ is the empty set and the statement holds. Now assume that $I(p_2(g)) \not\subseteq \{ 0 \}$. Hence, a number $a \in I(p_2(g)) \setminus \{ 0 \}$ exists. For each $z \geq n_g$, $za \geq n_g = n_{p_1(g)}$. According to the induction hypotheses $za \in I(p_1(g))$ and $z = za/a \in I(p_1(g)) / I(p_2(g)) = I(g)$. \qed
\end{proof}

\begin{corollary} \label{Folg1}
Let $C$ be a $\cO$-circuit with $ \emptyset \neq \cO \subseteq \{ \cup, \cap , \compl, +, / \}$. Then for all $z \geq 2^{|C|} + 1$,
$ z \in I(C) \; \Leftrightarrow \; 2^{|C|} + 1 \in I(C)$.
\end{corollary}

\begin{makethm}{theorem}{Satz1}
$\MC(\{ \cup, \cap , \compl, +, / \}) \in \PSPACE$.
\end{makethm}

\begin{proof} According to De Morgan's law, $MC( \cup, \cap, \compl , +, /) \logEquiv MC( \cup, \compl , +, /)$. Thus, it is sufficient to solve the membership problem for $\{ \cup, \compl, +, / \}$-circuits.
The membership problem for $\{ \cup, \compl, +, / \}$-circuits can be solved on an alternating \TM using a generalization of an algorithm by McKenzie and Wagner \cite[Lem. 5.4]{MW07}. We will answer the question of whether a natural number $x$ is contained in $I(g)$ for a gate $g$ by guessing from which numbers $x_1$ and $x_2$ the number $x$ was calculated and checking $x_1 \in I(p_1(g))$ and $x_2 \in I(p_2(g))$ in parallel using a universal state. Because of the $\compl$-gates we might need to check $x \notin I(g)$ instead. We do this by checking for all pairs of numbers $(x_1,x_2)$ that the number $x$ could have been calculated from $x_1 \notin I(p_1(g))$ or $x_2 \notin I(p_2(g))$. In each case, we use \Cref{Folg1} to limit the size of the values $x$ to $O(|C|)$ bits: If a value $x$ is too large, we can check $2^{|C|} + 1 \in I(g)$ instead. We get a \polyTime algorithm for alternating Turing machines.
This can be simulated in \polySpace on a ``classical'' \TM and we get $\MC(\cO) \in \PSPACE$.

\begin{figure}
\algrenewcommand\algorithmicrequire{\textbf{Input:}}
\algrenewcommand{\algorithmiccomment}[1]{\State {\color{blue}// #1 }}
\begin{algorithmic}[1]
   \Require $(C,b)$
   \State Let $m = 2^{|C|} + 1$, $g = g_C$, $x = b$ and $t = 1$.
   \While{$\alpha(g) \notin \mathbb{N}$}
      \If{$x > m $}
   		\State Let $x = m$. \label{L_m}
      \EndIf
      \Comment{If $t==1$ test $x \in I(g)$ and if $t==0$ test $x \notin I(g)$.}
      \If{$\alpha(g) == \compl $} 
        \State Let $g = p(g)$ and $t = 1-t$. \label{L_Komplement}
      
      \ElsIf{$\alpha(g) == \cup $ and $t == 0$}
      	\State Set $g = p_1(g)$ and $g = p_2(g)$ in parallel using an universal state. \label{L_Vereinigung0}       
 
      \ElsIf{$\alpha(g) == \cup $ and $t == 1$}
      	\State Guess $i \in \{ 1,2 \}$. \label{L_Vereinigung1}   
        \State Let $g = p_i(g)$.

      \ElsIf{$\alpha(g) == + $ and $t == 1$} 
        \State Bitwise guess a binary encoded $a \in [0, x] \cap \mathbb{N}$. \label{L_Addition1}
        \State Set  $(x,g) = (x-a, p_1(g))$ and $(x,g) = (a, p_2(g))$ in parallel using an universal state.
      \ElsIf{$\alpha(g) == + $ and $t == 0$}
        \State Set $a$ to every value from $[0, x] \cap \mathbb{N}$ in parallel using universal states. 
        \State Guess $i \in \{ 0, 1 \}$.
        \If{$i == 0$}
          \State Let $(x,g) = (x-a, p_1(g))$.
        \Else
          \State Let $(x,g) = (a, p_2(g))$.
        \EndIf  
      
      \ElsIf{$\alpha(g) == / $ and $t == 1$}
      	\State Bitwise guess a binary encoded $a \in [1, m] \cap \mathbb{N}$. \label{L_Division1} 
        \State Set $(x,g) = (a \cdot x, p_1(g))$ and  $(x,g) = (a, p_2(g))$ in parallel using an universal state.
      \ElsIf{$\alpha(g) == / $ and $t == 0$}
        \State Set $a$ to every value from  $[1, m] \cap \mathbb{N}$ in parallel using universal states. 
        \State Guess $i \in \{ 1,2 \}$.
        \If{$i == 0$}
          \State Let $(x,g) = (a \cdot x, p_1(g))$      		
        \Else
          \State Let $(x,g) = (a, p_2(g))$.
        \EndIf
      \EndIf
   \EndWhile
   \Comment{$g$ is an input gate.} \label{L_Blatt}
   \If{$t == 1$}
      \State Accept if $x == \alpha(g)$. Otherwise, reject.
   \Else
      \State Reject if $x == \alpha(g)$. Otherwise accept.
   \EndIf
\end{algorithmic}
\caption{Algorithm for circuits with union, complement, addition and division}
\label{Algo1}
\end{figure}

\Cref{Algo1} shows the algorithm that solves the problem $MC( \cup, \compl , +, /)$ for valid inputs $(C,b)$ on an alternating Turing machine. We first want to show the correctness of the algorithm. Let $C$ be a $\{ \cup, \compl , +, / \}$-circuit and $b \in \natNum$. Let $\beta_M(C,b)$ be the computation tree of the algorithm for the input $(C,b)$. We want to show that $\beta_M(C,b)$ has an accepting subtree \ifAndOnlyIf $b \in I(C)$. For a gate $g'$ of $C$, $x' \in \natNum$ and $t' \in \{ 0, 1 \}$, let $\beta(g',x',t' )$ be a subtree of $\beta_M(C,b)$ from a point in time at which the variables $g$, $x$ and $t$ have taken on the values $g'$, $x'$ and $t'$ . $\beta(g',x',t')$ can occur multiple times in $\beta_M(C,b)$. $\beta(g',x',t')$ is still unique because all occurrences are identical.
Also define \\
$$A(g',x',t') \defEquals \left\{
\begin{array}{ll}
1 & \textrm{if } x' \in I(g') \textrm{ and } t'=1 \\
1 & \textrm{if } x' \notin I(g') \textrm{ and } t'=0 \\
0 & \textrm{otherwise} \\
\end{array}
\right.$$
We show inductively via the structure of $C$ that $A(g',x',t') = 1$ \ifAndOnlyIf $\beta(g',x',t')$ has an accepting subtree. \\
\underline{Base case:} Let $g'$ be an input gate. Then the algorithm will reach Line \ref{L_Blatt} and accept exactly under the conditions under which $A(g',x',t') = 1$ holds. \\
\underline{Induction step:} Let $g'$ be a gate of $C$ with predecessors. We assume that the statement holds for all $x \in \natNum$ and for all $t \in \{ 0, 1 \}$ and $p(g' )$ or, respectively, $p_1(g')$ and $p_2(g')$ (induction hypothesis). Let $x' \in \natNum$ and $t \in \{ 0, 1 \}$. W.l.o.g. let $x' \leq m = 2^{|C|} + 1$. Otherwise, the value $m$ is assigned to $x$ in Line \ref{L_m}. By \Cref{Folg1}:
$$ x' \in I(g') \Leftrightarrow m \in I(g')$$
Hence if $x' > m$, we get $A(g', x', t') = A(g', m, t')$ by \Cref{Folg1}. The algorithm does not branch in this line. Therefore, $\beta(g', m, t')$ is a subtree of $\beta(g', x', t')$ that has an accepting subtree \ifAndOnlyIf $\beta(g', x', t')$ has one.

\emph{Case 1:} We have $\alpha(g') = \compl$. The algorithm sets $(g,x,t) = (p(g'), x',1-t')$ in Line \ref{L_Komplement} without branches. Therefore, $\beta(g',x',t')$ has an accepting subtree \ifAndOnlyIf $\beta(p(g'),x',1-t')$ has an accepting subtree. The following applies:
\begin{align*}
A(g',x',t') & = \left\{
\begin{array}{ll}
1 & \textrm{if } x' \in I(g') \textrm{ and } t'=1 \\
1 & \textrm{if } x' \notin I(g') \textrm{ and } t'=0 \\
0 & \textrm{otherwise} \\
\end{array}
\right. \\
& = \left\{
\begin{array}{ll}
1 & \textrm{if } x' \notin I(p(g')) \textrm{ and } 1-t'=0 \\
1 & \textrm{if } x' \in I(p(g')) \textrm{ and } 1-t'=1 \\
0 & \textrm{otherwise} \\
\end{array}
\right. \\
& = A(p(g'), x', 1-t')
\end{align*}
By the induction hypothesis, $\beta(p(g'),x',t')$ has an accepting subtree \ifAndOnlyIf $A(p(g'),x',t') = 1$. Therefore, $\beta(g',x',t')$ has an accepting subtree \ifAndOnlyIf $A(g',x',t') = 1$.

\emph{Case 2:} We have $\alpha(g') = \cup $ and $t' = 0$. The algorithm reaches Line \ref{L_Vereinigung0} and is in a universal state. $(g,x,t)$ is set from $(g',x',t')$ to $(p_1(g'),x',t')$ and $(p_2(g'),x ',t')$ in parallel. Hence $\beta(g',x',t')$ has an accepting subtree \ifAndOnlyIf $\beta(p_1(g'),x',t')$ and $\beta(p_2(g') 'x',t')$ each have an accepting subtree. The following applies:
\begin{align*}
A(g',x',t') & = \left\{
\begin{array}{ll}
1 & \textrm{if } x' \notin I(g')\\
0 & \textrm{otherwise} \\
\end{array}
\right. \\
& = \left\{
\begin{array}{ll}
1 & \textrm{if } x' \notin I(p_1(g')) \textrm{ and } x' \notin I(p_2(g'))\\
0 & \textrm{otherwise} \\
\end{array}
\right. \\
& = \left\{
\begin{array}{ll}
1 & \textrm{if } A(p_1(g'),x',t')=1 \textrm{ and } A(p_2(g')',x',t')=1\\
0 & \textrm{otherwise} \\
\end{array}
\right. 
\end{align*}
From the induction hypothesis, it follows that $\beta(g',x',t')$ has an accepting subtree \ifAndOnlyIf $A(g',x',t') = 1 $.

\emph{Case 3:} We have $\alpha(g') = \cup $ and $t' = 1$. The algorithm reaches Line \ref{L_Vereinigung1} and is in an existential state. $(g,x,t)$ is set from $(g',x',t')$ to $(p_1(g'),x',t')$ and $(p_2(g'),x ',t')$ in parallel. $\beta(g',x',t')$ has an accepting subtree \ifAndOnlyIf $\beta(p_1(g'),x',t')$ or $\beta(p_2(g'), x',t')$ has an accepting subtree. The following applies:
\begin{align*}
A(g',x',t') & = \left\{
\begin{array}{ll}
1 & \textrm{if } x' \in I(g')\\
0 & \textrm{otherwise} \\
\end{array}
\right. \\
& = \left\{
\begin{array}{ll}
1 & \textrm{if } x' \in I(p_1(g')) \textrm{ or } x' \in I(p_2(g'))\\
0 & \textrm{otherwise} \\
\end{array}
\right. \\
& = \left\{
\begin{array}{ll}
1 & \textrm{if } A(p_1(g'),x',t')=1 \textrm{ or } A(p_2(g')',x',t')=1\\
0 & \textrm{otherwise} \\
\end{array}
\right. 
\end{align*}

From the induction hypothesis, it follows that $\beta(g',x',t')$ has an accepting subtree \ifAndOnlyIf $A(g',x',t') = 1 $.

\emph{Case 4:} We have $\alpha(g') = + $ and $t' = 1$. The algorithm reaches Line \ref{L_Addition1}. For each $a \in [0, x'] \cap \mathbb {N}$, $\beta(g',x',t')$ has a subtree $\beta(p_1(g'),x'- a,t')$ and a subtree $\beta(p_2(g'),a,t')$, which are reached by a chain of existential states and a universal state. Hence $\beta(g',x',t')$ has an accepting subtree \ifAndOnlyIf there is a $a \in [0, x'] \cap \mathbb {N}$ such that $\beta( p_1(g'),x'-a,t')$ and $\beta(p_2(g'),a,t')$ have an accepting subtree.

\begin{align*}
A(g',x',t') & = \left\{
\begin{array}{ll}
1 & \textrm{if } x' \in I(g')\\
0 & \textrm{otherwise} \\
\end{array}
\right. \\
& = \left\{
\begin{array}{ll}
1 & \textrm{if } \exists a \in \natNum: x'-a \in I(p_1(g')) \textrm{ and } a \in I(p_2(g'))\\
0 & \textrm{otherwise} \\
\end{array}
\right. \\
& = \left\{
\begin{array}{ll}
1 & \textrm{if } \exists a \in [0, x'] \cap \natNum: x'-a \in I(p_1(g')) \textrm{ and } a \in I(p_2(g'))\\
0 & \textrm{otherwise} \\
\end{array}
\right. \\
& = \left\{
\begin{array}{ll}
1 & \textrm{if } \exists a \in [0, x'] \cap \natNum: \\ 
  & A(p_1(g'), x'-a, t') = 1 \textrm{ and } A(p_2(g'), a, t') = 1 \\
0 & \textrm{otherwise} \\
\end{array}
\right.
\end{align*}

According to the induction hypothesis, $\beta(g',x',t')$ has an accepting subtree \ifAndOnlyIf $A(g',x',t') = 1 $.

\emph{Case 5:}
We have $\alpha(g') = + $ and $t' = 0$. The statement can be shown analogously to Case 4. Universal and existential states, the quantifier and ``and'' and ``or'' are each swapped.

\emph{Case 6:} We have $\alpha(g') = / $ and $t' = 1$. The algorithm reaches Line \ref{L_Division1}. For each $a \in [1, m] \cap \mathbb {N}$, $\beta(g',x',t')$ has a subtree $\beta(p_1(g'),a \cdot x 't')$ and a subtree $\beta(p_2(g'),a,t')$, which are reached through a chain of existential states and a universal state. Hence $\beta(g',x',t')$ has an accepting subtree \ifAndOnlyIf there is an $a \in [1, m] \cap \mathbb {N}$ such that $\beta(p_1 (g'),a \cdot x',t')$ and $\beta(p_2(g'),a,t')$ have an accepting subtree.

\begin{align*}
A(g',x',t') & = \left\{
\begin{array}{ll}
1 & \textrm{if } x' \in I(g')\\
0 & \textrm{otherwise} \\
\end{array}
\right. \\
& = \left\{
\begin{array}{ll}
1 & \textrm{if } \exists a \in \natNum^+:  a \cdot x' \in I(p_1(g')) \textrm{ and } a \in I(p_2(g'))\\
0 & \textrm{otherwise} \\
\end{array}
\right. 
\end{align*}

By \Cref{Folg1} the following applies for all $a > m$:
$$ a \cdot x' \in I(p_1(g')) \Leftrightarrow  m \cdot x' \in I(p_1(g')) $$ 
and 
$$ a \in I(p_2(g')) \Leftrightarrow  m \in I(p_2(g')) $$

That is why:
\begin{align*}
A(g',x',t') & = \left\{
\begin{array}{ll}
1 & \textrm{if } \exists a \in [1, m] \cap \natNum:  a \cdot x' \in I(p_1(g')) \textrm{ and } a \in I(p_2(g'))\\
0 & \textrm{otherwise} \\
\end{array}
\right. \\
& = \left\{
\begin{array}{ll}
1 & \textrm{if } \exists a \in [1, m] \cap \natNum: \\
  &  A(g', a \cdot x,t') = 1 \textrm{ and }  A(g', a ,t') = 1 \\
0 & \textrm{otherwise} \\
\end{array}
\right. 
\end{align*}

Therefore, by the induction hypothesis, $\beta(g',x',t')$ has an accepting subtree \ifAndOnlyIf $A(g',x',t') = 1 $.

\emph{Case 7:} We have $\alpha(g') = / $ and $t' = 0$. The statement can be shown analogously to Case 6.

In each case, the computation tree $\beta(g_C,b,1)$ has an accepting subtree \ifAndOnlyIf $A(g_C,b,1) = 1$. This is the case \ifAndOnlyIf $b \in I(g_C) = I(C)$. Thus, we have proven that the algorithm is correct.

It remains to show that the alternating algorithm only requires \polyTime. The variable $g$ is assigned the value of each gate from $C$ at most once on every path from the root of $\beta_M(C,b)$ to a leaf, because $C$ is acyclic. In each iteration of the loop, $g$ changes. Therefore, the number of loop iterations is in $O(|C|)$. The running time of one iteration of the loop is $O(|C|^2)$. Therefore, the running time of \Cref{Algo1} is in $O(n^3)$, where $n$ is the size of the input $(C,b)$.
It is known that $\ATIME(n^3) \subseteq \DSPACE(n^3)$. The validity of the input can be checked in deterministic \logSpace.
Hence $MC( \cup, \compl , +, /) \in \DSPACE(n^3) \subseteq \PSPACE $. \qed
\end{proof}

\noindent Combined with results by McKenzie and Wagner~\cite[Lem. 5.1]{MW07} we get: 
\begin{corollary} \label{Koro1}
$\MC( \cup, \cap , \compl, +, / )$ and $\MC( \cup, \cap, +, / )$ are $\PSPACE$-complete.
\end{corollary}

\subsubsection{No prime numbers without multiplication}
From \Cref{Folg1} we get another result: McKenzie and Wagner described how the set of prime numbers $\primes$ is represented by a $\{ \cup, \cap , \compl, \times \}$-circuit and were thus able to express Goldbach's conjecture as a membership problem \cite{MW07}. We can prove that $\{ \cup, \cap , \compl, +, / \}$-circuits are not suitable for expressing the prime numbers. First, we need the following statement about the structure of the sets described by $\{ \cup, \cap , \compl, +, / \}$-circuits: 
\begin{corollary}\label{Kor1}
Let $A \subseteq \natNum$ be such that $A$ and $\overline{A}$ are infinite. \\
There is no $\{ \cup, \cap , \compl, +, / \}$-circuit $C$ with $I(C) = A$. 
\end{corollary} 
\begin{proof}
Suppose that there is such a $\{ \cup, \cap , \compl, +, / \}$-circuit $C$. $A$ and $\overline{A}$ are infinite. Hence, there is $a \in A$ with $a \geq 2^{|C|} + 1$ and $b \in \overline{A}$ with $b \geq 2^{|C|} + 1$ .
Applying \Cref{Folg1} twice we get:
$$a \in A = I(C) \quad \Rightarrow \quad 2^{|C|} + 1 \in I(C) \quad \Rightarrow \quad b \in I(C) = A$$
This is a contradiction to $b \in \overline{A}$. Hence, such a circuit $C$ does not exist. \qed
\end{proof}

\noindent Since both $\primes$ and $\primes \setminus \natNum$ are infinite, it follows:
\begin{corollary} \label{Sequence6}
There is no $\{ \cup, \cap , \compl, +, / \}$-circuit $C$ with $I(C) = \primes$.
\end{corollary}

\subsection{Circuits without complement}
The sets calculated by the circuits without complement are finite and upper-bounded. The upper bounds are as follows: 

\begin{makethm}{lemma}{Lemma3}
(i) For each $\{ \cup, \cap , +, / \}$-circuit $C$, $I(C) \subseteq \left[ 0, 2^{|C|} \right]$. \\
(ii) For each $\{ \cup, \cap , +, \times, / \}$-circuit $C$, $I(C) \subseteq \left[ 0, 2^{2^{|C|}} \right]$.
\end{makethm}

\begin{proof}
\begin{enumerate} 
\item[(i)] $I(C)$ is finite. The claim follows directly from \Cref{Folg1}. 

\item[(ii)] Let $C_g$ be the circuit formed from $C$ using $g$ as the output gate and removing all gates from which $g$ cannot be reached. We show by structural induction that for all gates $g$, $I(g) \subseteq \left[ 0, 2^{2^{|C_g|}} \right]$. 

Let $C$ be a $\{ \cup, \cap , +, \times, / \}$-circuit. We show $I(C) \subseteq \left[ 0, 2^{2^{|C|}} \right]$ by structural induction. Let $C_g$ be the circuit formed from $C$ by using $g$ as the output gate and removing all gates from which $g$ cannot be reached.

\underline{Base case:} Let $g$ be an input gate. Hence $I(C_g) = \{ \alpha(g) \}$.  The binary encoding of $\alpha(g)$ 
is contained in the encoding of $C_g$. This shows $\alpha(g) \leq 2^{|C_g|} \leq 2^{2^{|C_g|}}$.

\underline{Induction step:} Let $g$ be a gate with predecessors such that $$I(p_i(g)) \subseteq \left[ 0, 2^{2^{\left| C_{p_i(g)} \right|}} \right]$$ holds for all $i \in \{ 1,2 \}$. W.l.o.g. let $max(I(p_1(g))) \geq max(I(p_2(g)))$. Then we get:
{ \small $$
max(I(g)) \leq (max(I(p_1(g))))^2 +1
          \leq \left(  2^{2^{\left| C_{p_1(g)} \right| }} \right)^2 +1 
          = 2^{2^{\left( \left| C_{p_1(g)} \right| +1 \right) }} +1
          \leq 2^{2^{ \left| C_{p(g)} \right|}}
$$ } \qed
\end{enumerate}
\end{proof}

\noindent In circuits without complement and union, the sets calculated for the gates always contain at most one element. 

\begin{theorem} \label{Satz3}
(i) $\MC(\{ \cap, +, / \}) \in \mathrm{P}$ and (ii) $\MC(\{ \cap, +, \times, / \}) \in \E$.
\end{theorem}

\begin{proof}
(i) Let $b \in \natNum$ and let $C$ be an $\{ \cap, +, / \}$-circuit with output gate $g_C$. We want to check $(C,b) \in \MC(\mathcal {O})$. For each gate $g$ from $C$, $|I(g)| \leq 1$ and for all $n \in I(g)$, $n \leq 2^{|C|}$ by \Cref{Lemma3} (i). Hence $I(g)$ can be stored in $|C|$ bits. Compute $I(C)$ iteratively by computing $I(g)$ for each gate $g$. In each iteration, find a gate $g$ with predecessors $g_0$ and $g_1$ such that $I(g_0)$ and $I(g_1)$ are already calculated and calculate $I(g)$. $I(g)$ can be calculated in polynomial time. Repeat this until $I(g_C)$ is calculated. Since the number of gates is upper-bounded by $|C|$, only \polyTime is needed to calculate the set $I(g_C)$. Now check $b \in I(g_C) = I(C)$. \\
(ii) Analogous to (i) with \Cref{Lemma3} (ii). \qed
\end{proof}

\begin{remark}
\Cref{Satz3} shows that $\MC(\{ \cap, +, \times, / \})$ is in the class $\EXP$ that is closed under log-space many-one reduction.
\end{remark}

In circuits with $\cup$-gates the sets of the gates may contain multiple elements each. We can guess the representatives that prove $(C,b) \in \MC( \cup , \cap , +, \times , / )$ by nondeterminism, but in order to calculate a number $b$ in the output gate several different values may be required from each gate of the circuit. 
\begin{theorem} \label{Satz6}
$ \MC( \cup , \cap , +, \times , / ) \in \NEXP $.
\end{theorem}
\begin{proof}
For this proof, let $\cO \defEquals \{ \cup , \cap , +, \times , / \}$.
An $\cO$-formula is a $\cO$-circuit such that the gates have a maximum output degree of $1$. First we consider the membership problem $\mathrm{MF}(\mathcal {O})$ for $\cO$-Formulas.
We first want to show that $\mathrm{MF}(\cO) \in \NP$.
Let $F=(G,E,g_F, \alpha )$ be a $\cO$-formula and $b \in \natNum$. Let $I$ be the set of labels of 
input gates of $F$. Then $m \defEquals \prod_{a \in I} (a+1)$ is an upper bound for all $I(g)$ such that $g$ is a gate from $F$.
For each gate $g$ from $F$, guess a number $b_g \leq m$ and check, if the guessed numbers prove $b \in I(g_F)$. Now let $C$ be an $\cO$-circuit and $b' \in \natNum$. To decide $(C,b') \in \MC( \mathcal {O} )$, we expand $C$ into an equivalent $\mathcal {O}$-formula $F_C$. (Similar to \cite[Thm. 6.2]{MW07}.) Then we use the $\NP$-algorithm above to test whether $(F,b') \in \mathrm{MF}( \cO )$. Since $F$ is exponentially larger than $C$ at maximum, this yields a $\NEXP$-algorithm for $ \MC( \cO )$. \qed
\end{proof}
\noindent McKenzie and Wagner proved that $\MC( \cup , \cap, +, \times )$ is $\NEXP$-complete \cite[Lem. 6.1 and Thm. 6.2]{MW07}. With \Cref{Satz6} we get:
\begin{corollary} 
\begin{enumerate}
\item[(i)] $\MC( \cup , \cap , +, \times , / )$ is $\NEXP$-complete.
\item[(ii)] $\MC( \cup , \cap , +, \times , / ) \logEquiv \MC( \cup , \cap , +, \times)$.
\end{enumerate}
\end{corollary}

\subsubsection{Considerations for $\MC(\cap, +, \times, /)$ using alternating Turing machines} 1212
While creating this work, there was an approach to solve $\MC(\cap, +, \times, /)$ in alternating polynomial time, which did not lead to success: With an alternating \TM it is possible for a $\{ \cap, +, \times, / \}$-circuit $C$ to guess prime numbers $p$ up to a size of $O(2^{2|C|})$. The result sets of the gates are each stored modulo $p$. $O(2|C|)$ bits per gate are sufficient for this. The \TM should check $ (f \text{ mod } p) \in \{ e \text{ mod } p \mid e \in I(g) \}$ instead of $ f \in I(g)$ for all such prime numbers in parallel and only accept if the tests for all such prime numbers are successful. If this is possible, the Chinese Remainder Theorem can be used to prove that there is only one number $b \leq 2^{2^{|C|}}$ for which the algorithm accepts $(C,b)$. This approach has two problems for $/$-gates:

\begin{enumerate}
\item The alternating \TM cannot check modulo $p$ whether for two numbers $x_1$ and $x_2$, both divisible by $p$, $x_1$ is divisible by $x_2$. The result of $x_1 / x_2$ modulo $p$ is undefined in this case, because $x_1 = x_2 \times a$ applies modulo $p$ for all natural numbers $a$. This problem can be solved by accepting for all prime numbers that appear as prime factors in the circuit. Using the Chinese Remainder Theorem it can be shown that the uniqueness of $b$ with $I(C) = \{b \}$ is guaranteed if $b$ modulo $p$ is given for $2^{|C|}$ prime numbers. This is guaranteed to be the case because we guess prime numbers up to size $O(2^{2|C|})$.
\item The integers modulo $p$ corresponding to the finite field $\mathbb{F}_p$. In this field, every division that is not a division by zero of the field has a result. We consider the following example: 
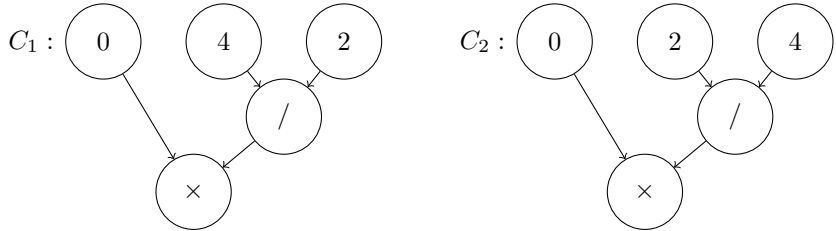
\begin{center}
\hspace{-0.75cm}
\begin{tikzpicture}
	\newcommand\dx{0.8}
	\newcommand\dy{1}
	\newcommand\minSize{1cm}
	\newcommand\xoffset{7.5*\dx}
	
	\node[fill=white,radius=0.5, minimum size=1.4cm] (labelA) at (-0.2*\dx, -0*\dy) {$C_1:$};
	
	\node[draw,circle,fill=white,radius=0.5, minimum size=\minSize] (inAa) at (1*\dx, -0*\dy) {$0$};
	\node[draw,circle,fill=white,radius=0.5, minimum size=\minSize] (inAb) at (3*\dx, -0*\dy) {$4$};
	\node[draw,circle,fill=white,radius=0.5, minimum size=\minSize] (inAc) at (5*\dx, -0*\dy) {$2$};
	
	\node[draw,circle,fill=white,radius=0.5, minimum size=\minSize] (divA) at (4*\dx, -1*\dy) {$/$};
	
	\node[draw,circle,fill=white,radius=0.5, minimum size=\minSize] (multA) at (2.5*\dx, -2*\dy) {$\times$};
	
	\draw[->] (inAb) to (divA);
	\draw[->] (inAc) to (divA);
	
	\draw[->] (inAa) to (multA);
	\draw[->] (divA) to (multA);

	\node[fill=white,radius=0.5, minimum size=1.4cm] (labelB) at (\xoffset - 0.2*\dx, -0*\dy) {$C_2:$};
	
	\node[draw,circle,fill=white,radius=0.5, minimum size=\minSize] (inBa) at (\xoffset + 1*\dx, -0*\dy) {$0$};
	\node[draw,circle,fill=white,radius=0.5, minimum size=\minSize] (inBb) at (\xoffset + 3*\dx, -0*\dy) {$2$};
	\node[draw,circle,fill=white,radius=0.5, minimum size=\minSize] (inBc) at (\xoffset + 5*\dx, -0*\dy) {$4$};
	
	\node[draw,circle,fill=white,radius=0.5, minimum size=\minSize] (divB) at (\xoffset + 4*\dx, -1*\dy) {$/$};
	
	\node[draw,circle,fill=white,radius=0.5, minimum size=\minSize] (multB) at (\xoffset + 2.5*\dx, -2*\dy) {$\times$};
	
	\draw[->] (inBb) to (divB);
	\draw[->] (inBc) to (divB);
	
	\draw[->] (inBa) to (multB);
	\draw[->] (divB) to (multB);
\end{tikzpicture}
\end{center}
We have $I(C_1) = \{ 0 \}$ and $I(C_2) = \emptyset$. The alternating \TM can find a number $x_p$ for every prime number $p > 2$, so that in the field $\mathbb{F}_p$ the equation $2 / 4 = x_p$ holds. For all such $p$, $0 \times x_p = 0$. Therefore, the alternating \TM would accept $(C_2, 0)$ even though $I(C_2) = \emptyset$.
\end{enumerate}

\subsection{Circuits without addition} 
We extend the proofs for $\MC(\cap, \times) \in \mathrm{P}$ and $\MC(\cup, \cap , \compl, \times) \in \PSPACE$ by McKenzie and Wagner \cite[Section 3]{MW07} to prove $\MC(\cap, \times, /) \in \mathrm{P}$ as well as $\MC(\cup, \cap , \compl, \times, /) \in \PSPACE$. 

\begin{definition}
Let $a_1, \dots, a_n \in \natNum \setminus \{ 0 \}$. Then natural numbers $k \geq 1$, $q_1, \dots q_k \geq 2$ and $e_{1,1}, \dots e_{n,k} \geq 0$ exist such that $gcd(q_i, q_j) = 1$ for all $i,j \in \{1, \dots, k \}$ with $i \neq j$ and $a_i = \prod_{j=1}^k q_j^{e_{i,j }}$ for all $i \in \{1, \dots, n \}$. $\{ q_1, \dots q_k \}$ is called \emph{GCD-free basis} of $\{ a_1, \dots, a_n \}$. 
\end{definition} 

A GCD-free basis can be calculated in polynomial time \cite[Prop. 3.1]{MW07}. To solve the membership problem with input $(C,b)$, we will convert the numbers in the circuit $C$ to their factorization with respect to the GCD-free base (represented as a vector of powers) of all the numbers from the input gates and the number $b$. This leads to the following definition: 
\begin{definition}
Let $\emptyset \neq \cO \subseteq \{ \cup, \cap, \compl, +, -, \}$. $\MC^{*}(\cO)$ is the membership problem for $\cO$-circuits over $\natNumInf^k \defEquals \natNum^k \cup \{ \infty \}$ for $k \geq 1$. The arithmetic operations on $\natNum^k$ are defined component-wise. The vector $(a_1, \dots, a_k) - (b_1, \dots, b_k)$ is defined \ifAndOnlyIf $a_i - b_i$ is a natural number for all $i \in [k]$ . For all $m \in \mathbb {N}^k$, $\infty + m = \infty = m + \infty$, $\infty + \infty = \infty$, $\infty - m = \infty$ and $m - \infty$ as well as $\infty - \infty$ are undefined.
\end{definition} 
\begin{remark}
For each $k \geq 1$, $(\natNumInf^k, +, (0, \dots, 0))$ is a monoid.
\end{remark}

\begin{theorem} \label{Satz9}
For each $\cO$ such that $\emptyset \neq \cO \subseteq \{ \cup, \cap \}$, 
\begin{enumerate}
\item[(i)] $\MC(\cO \cup \{ \times, / \}) \pReduction \MC^*(\cO \cup \{ +, - \})$ and 
\item[(ii)] $\MC(\cO \cup \{  \compl, \times, / \}) \pspaceReduction \MC^*(\cO \cup \{ \compl, +, - \})$.    
\end{enumerate}
\end{theorem}polynomial time

\begin{proof} 
(i) Let $\cO$ such that $\emptyset \neq \cO \subseteq \{ \cup, \cap \}$ and $C$ be a $\cO \cup \{ \times, / \}$-circuit with input gates $g_1, .. ., g_n$ with labels $a_1, \dots, a_n \in \natNum$ and let $b \in \natNum$. Compute a GCD-free basis $\{q_1, \dots, q_k \}$ of $\{ a_1, \dots, a_n, b \}$ in polynomial time. Let $M \defEquals \left\{ \prod_{j=1}^k q_j^{d_j} \mid d_1, \dots, d_k \in \natNum \right\}$. Hence $M \subseteq \mathbb {N}^+$, $M \cup \{ 0 \}$ is closed under multiplication and $1 = \prod_{j=1}^k q_j^0 \in M$. Therefore, $(M \cup \{ 0 \}, \times, 1)$ (such that $\times$ is the multiplication of $\natNum$ restricted to $M \cup \{ 0 \}$) is a monoid. Define $\sigma: M \cup \{ 0 \} \rightarrow \natNumInf^k$ such that $\sigma(\prod_{j=1}^k q_j^{d_j}) = (d_1, \dots, d_k)$ and $\sigma(0) = \infty$. The mapping $\sigma$ is a monoid isomorphism. Let $C'$ be the $\cO \cup \{ +, - \}$-circuit, which we construct from $C$ by labeling the input gate $g_i$, for $i \in \{ 1, \dots n \}$, with $\sigma(a_i)$ instead of $a_i$, replacing all $\times$-gates with $+$-gates, and replacing all $/$-gates with $-$-gates. For all $S \subseteq M \cup \{ 0 \}$ let $\sigma(S)$ be the image of $S$ under $\sigma$. For all $S_1, S_2 \subseteq M \cup \{ 0 \}$, the following holds:
{\small \begin{align*}
\sigma(S_1 \times S_2) & = \sigma(S_1) + \sigma(S_2) & \sigma(S_1 / S_2) & = \sigma(S_1) - \sigma(S_2) \\
\sigma(S_1 \cup S_2) & = \sigma(S_1) \cup \sigma(S_2) & \sigma(S_1 \cap S_2) & = \sigma(S_1) \cap \sigma(S_2)
\end{align*} }
With these equations we can show: 
\begin{cclaim} \label{auxStatment1}
For all gates $g$ from $C$ and $a \in \natNum$, $a \in I_C(g) \Leftrightarrow \sigma(a ) \in I_{C'}(g)$. 
\end{cclaim} 
According to Claim 1 for all $a \in M \cup \{ 0 \}$, $a \in I_C(g)$ \ifAndOnlyIf $\sigma(a) \in I_{C'}(g)$.
Since $b \in M \cup \{ 0 \}$ we get $(C,b) \in \MC(\cO \cup \{ \times, / \})$ \ifAndOnlyIf $(C', \sigma(b)) \in \MC^*(\cO \cup \{ +, - \})$.

We will now show Claim 1. First, we define: 
\begin{align*}
M & \defEquals \left\{ \; \prod_{j = 1}^k q_j^{d_j} \mid d_1, \dots, d_k \in \natNum \right\} \\
K_i & \defEquals \left\{  \prod_{j = k+1}^\infty q_j^{d_j} \mid d_{k+1}, d_{k+2}, \dots \in \natNum \text{ and } \sum_{j = k+1}^\infty d_j = i \right\}  \text{ for each } i \in \natNum
\end{align*}
In the following we write the element-wise multiplication of the sets as $\times$ and the Cartesian product as $\otimes$. Hence $M = M \times \{ 1 \} = M \times K_0$ and $\natNum \setminus \{ 0 \} = \bigcup_{i = 0}^\infty M \times K_i$. Note that $\sigma(M \times K_i) = \natNum^k \otimes \{ i \}$ for all $i \in \natNum$. For all $S_1, S_2 \subseteq M$ the following applies: 
\begin{align*}
\sigma(S_1 \times S_2) & = \sigma(S_1) + \sigma(S_2) & 
\sigma(S_1 / S_2) & = \sigma(S_1) - \sigma(S_2) \\
\sigma(S_1 \cup S_2) & = \sigma(S_1) \cup \sigma(S_2) & 
\sigma(M \setminus S_1) & = \sigma(M) \setminus \sigma(S_1)
\end{align*}
For each $i \in \natNum$, Let $e_i \defEquals (0, \dots, 0, i) \in \natNum^{k+1}$. \\
\begin{cclaim} \label{auxStatment2}
For every gate $g$ in $C$, $T \subseteq \{ 0 \}$ and $S_0, S_1, \dots \subseteq M$ exist such that $I_C(g) = T \cup \bigcup_{i \in \natNum} (S_i \times K_i)$ and $I_{C'}(g) = \sigma(T) \cup \bigcup_{i \in \natNum} \left(  \sigma(S_i) + \{ e_i \} \right)$. 
\end{cclaim}
\Cref{auxStatment2} can be proven inductively via the structure of $C$. We will only provide the induction step for the case of a $/$-gate, because the rest is analogous to the proof of \cite[Lem. 3.2 (iii)]{MW07}. Let $g$ be a $/$-gate with first predecessor $g_1$ and second predecessor $g_2$. According to the induction hypothesis, $T_1, T_2 \subseteq \{ 0 \}$, $S^1_0, S^1_1, \dots \subseteq M$ and $S^2_0, S^2_1, \dots \subseteq M$ exist, such that $I_C(g_j) = T_j \cup \bigcup_{i \in \natNum} (S^j_i \times K_i) $ and $I_{C'}(g_j) = \sigma(T_j) \cup \bigcup_{i \in \natNum} \left(  \sigma(S^j_i) + \{ e_i \} \right) $ for each $j \in \{ 1, 2 \}$.
The following applies:
$$
I(g) 
= I(g_1) / \left(T_2 \cup \bigcup_{i \in \natNum} (S^2_i \times K_i)\right) 
= (I(g_1) / T_2) \cup \left(I(g_1) / \bigcup_{i \in \natNum} (S^2_i \times K_i)\right)
$$
Because of $T_2 \subseteq \{ 0 \}$ this shows:
$$
I(g) = \left(I(g_1) / \bigcup_{i \in \natNum} (S^2_i \times K_i)\right)
= \left(T_1 \cup \bigcup_{i \in \natNum} (S^1_i \times K_i)\right) / \bigcup_{i \in \natNum} (S^2_i \times K_i)
$$
Hence $ I(g) = T \cup \left(\left(\bigcup_{i \in \natNum} (S^1_i \times K_i)\right)/\left(\bigcup_{i \in \natNum} (S^2_i \times K_i)\right)\right) $ with \\
$ T \defEquals T_1 / \left(\bigcup_{i \in \natNum} (S^2_i \times K_i)\right) $. \\
\begin{cclaim} \label{auxStatment3}
For each $i \in \natNum$, $I(g) = T \cup \bigcup_{i \in \natNum} (S_i \times K_i)$ with $S_i \defEquals \bigcup_{j = i}^\infty (S^1_j / S^2_{j-i}) $. 
\end{cclaim}
Proof of \Cref{auxStatment3}: Let $x \in (\bigcup_{i \in \mathbb {N}} (S^1_i \times K_i))/(\bigcup_{i \in \mathbb {N}} (S^2_i \times K_i)) $. Then there are $i_1 \in \mathbb {N}$, $s_1 \in S_{i_1}^1$ and $k_1 \in K_{i_1}$, as well as $i_2 \in \mathbb {N}$, $s_2 \in S_{i_2}^2$ and $k_2 \in K_{i_2}$ with $x = (s_1 k_1)/ (s_2 k_2)$. $s_2$ divides $s_1$ and $k_2$ divides $k_1$. In particular, $i_1 \geq i_2$. Hence $x = (s_1/s_2) (k_1/k_2) \in (S^1_{i_1} / S^2_{i_2}) \times K_{i_1 - i_2}$ and $S^1_{i_1} / S^2_{i_2} \subseteq S_{i_1 - i_2} $. Conversely, let $x \in \bigcup_{i \in \natNum} ((\bigcup_{j = i}^\infty (S^1_j / S^2_{j-i})) \times K_i)$. 
Then there are $i \in \natNum$, $j \geq i$, $s_1 \in S_j^1$, $s_2 \in S_{j-i}^2$ and $l \in K_i$ with 
 $ x = (s_1 / s_2) \; l = (s_1 l q_{k+1}^{j-i}) / (s_2 q_{k+1}^{j-i}) \in (S_j^1 \times K_j) / (S_{j-i}^2 \times K_{j-i})$. Thus, \Cref{auxStatment3} holds. 

For $I_{C'}(g)$ we get:
\begin{align*}
I_{C'}(g) 
  = &  I_{C'}(g_1) - \left(\sigma(T_2) \cup \bigcup_{i \in \natNum} (\sigma(S^2_i) + \{ e_i \})\right)  \\
= & I_{C'}(g_1) - \left(\bigcup_{i \in \natNum} (\sigma(S^2_i) + \{ e_i \})\right) \hspace{3cm} \text{ (because $\sigma(T_2) \subseteq \{ \infty \}$)} \\
= & \left(\sigma(T_1) \cup \bigcup_{i \in \natNum} (\sigma(S^1_i) + \{ e_i \})\right) - \left(\bigcup_{i \in \natNum} (\sigma(S^2_i) + \{ e_i \})\right) \\
= & \left(\sigma(T_1) - \left(\bigcup_{i \in \natNum} (\sigma(S^2_i) + \{ e_i \})\right)\right) \\
& \cup \left(\left(\bigcup_{i \in \natNum} (\sigma(S^1_i) + \{ e_i \})\right) - \left(\bigcup_{i \in \natNum} (\sigma(S^2_i) + \{ e_i \})\right)\right) \\
= & \sigma(T) \cup \sigma\left(\left(\bigcup_{i \in \natNum} (S^1_i \times K_i)\right)/\left(\bigcup_{i \in \natNum} (S^2_i \times K_i)\right)\right) \\
= & \sigma(T) \cup \sigma\left(\bigcup_{i \in \natNum} (S_i \times K_i)\right) \\
= & \sigma(T) \cup \bigcup_{i \in \natNum} (\sigma(S_i) + \{ e_i \} )
\end{align*}
This shows \Cref{auxStatment2}. We use \Cref{auxStatment2} to prove \Cref{auxStatment1}.

(ii) Now $\cO$ may include $\compl$. We can assume that $\cap \notin \cO$, because of De Morgan's laws. The $\compl$-gates can cause the circuit to produce natural numbers that can no longer be represented with respect to the GCD-free basis $q_1, \dots q_k$. We want to use the construction from (i) with the prime factorization (instead of the decomposition with respect to a GCD-free basis). In order to represent all the numbers that might be calculated in the gates of the circuit as a finite vector, we combine all powers of prime factors that are not prime factors in $a_1, \dots, a_n$ and $b$ into a single number.
Let $k \in \natNum$ such that $q_1 < \dots < q_k $ are the (ascendingly ordered) prime factors of the numbers $a_1, \dots, a_n$ and $b$. Let $q_{k+1}, q_{k+2}, \dots$ be all the other prime numbers (ascendingly ordered). We define $\sigma: \natNum \rightarrow \natNumInf^{k+1}$ such that $\sigma(0) = \infty$ and $\sigma(\prod_{j=1}^\infty q_j^{d_j}) = \left(  d_1, \dots, d_k, \sum_{j = k+1}^\infty d_j \right)$. For the reduction, we convert $C$ into a circuit $C'$ over $\natNumInf^k$ analogously to (i). The decomposition of $a_1, \dots, a_n$ and $b$ into prime factors is possible in $\PSPACE$. Therefore, $C'$ can be constructed in $\PSPACE$. Note that $\sigma$ is a monoid homomorphism between $(\natNum, \times, 1)$ and $(\natNumInf^{k+1}, +, (0, \dots 0))$. However, in contrast to the construction from (i), $\sigma$ is not injective. One can show (in a very technical) proof that for gates $g$ from $C$ and $a \in \natNum$, $a \in I_C(g) \Leftrightarrow \sigma(a ) \in I_{C'}(g)$ holds anyway. \qed 
\end{proof}

\noindent Analogously to \Cref{Satz3} (i) one shows:
\begin{theorem} \label{Satz10}
$\MC^*(\cap,  +, - ) \in \mathrm{P}$.
\end{theorem}

\begin{corollary} \label{Folg4}
$\MC(\cap,  \times, / ), \MC(\times, / ) \in \mathrm{P}$.
\end{corollary}
\noindent We limit the size of the elements for each node by generalizing \Cref{Lemma1}:
\begin{makethm}{lemma}{Lemma2}
Let $C = (V, E, g_C, \alpha)$ be a $\cO$-circuit with $\emptyset \neq \cO \subseteq \{ \cup, \cap , \compl, + , - \}$ over $\natNumInf^m $ with $m \geq 1$.
Let $n_C \defEquals (2^{|C|} + 1, \dots, 2^{|C|} + 1) \in \mathbb {N}^m$.
For each $g \in V$ and for each $z = (z_1, \dots, z_m) \in \mathbb {N}^m$ such that an index $1 \leq i \leq m$ exists with $z_i \geq 2^{|C|} + 1$, 
$z \in I(g)$ \ifAndOnlyIf $n_C \in I(g) $.
\end{makethm}

\begin{proof}
Let $C_g$ be the circuit formed from $C$ using $g \in V$ as the output gate and removing all gates from which $g$ cannot be reached. Let $n_g \defEquals (2^{|C_{g}|} + 1, \dots, 2^{|C_{g}|} + 1) \in \mathbb {N}^m$. \\
We show the following \textbf{auxiliary statement}: For each $g \in V$ and for each $z = (z_1, \dots, z_m) \in \mathbb {N}^m$ such that an index $1 \leq i \leq m$ exists with $z_i \geq 2^{|C_g|} + 1$, the following applies: 
$$ z \in I(g) \; \Leftrightarrow \; n_g \in I(g)$$
We show the auxiliary statement by induction on the structure of $C$. \\
\underline{Base case:} Let $g$ be an input gate and $\alpha(g) = (y_1, \dots, y_m)$. This means $I(g) = \{ (y_1, \dots, y_m) \}$. Since the binary representation of $y_i$ must be included in the encoding of $C_g$, $\log(y_i) + 1 \leq |C_g|$ and therefore $y_i \leq 2^{\log(y_i)+1} < 2^{|C_g|} + 1 $ for all $1 \leq i \leq m$. Hence, for each $z = (z_1, \dots, z_m) \in \mathbb {N}^m$ such that an index $1 \leq i \leq m$ exists with $z_i \geq 2^{|C_g|} + 1$, $z \notin I(g)$. In particular, $n_g \notin I(g)$. \\
\underline{Induction step:} Let $g$ be a gate of $C$ with predecessors such that the auxiliary statement holds for $p(g)$ or, respectively, $p_1(g)$ and $p_2(g)$ (induction hypothesis). Let $z = (z_1, \dots, z_m) \in \mathbb {N}^m$ and $z_i \geq 2^{|C_g|} + 1$ for an index i, such that $1 \leq i \leq m$. \\
\textit{Case 1:} Let $\alpha(g) = \cup$. It is $z_i \geq 2^{|C_g|} + 1 \geq 2^{|C_{p_j(g)}|} + 1$ for all $j \in \{1,2 \} $. It follows:
\begin{align*} 
z \in I(g) & \Leftrightarrow z \in I(p_1(g)) \vee z \in I(p_2(g)) \\
& \Leftrightarrow n_{p_1(g)} \in I(p_1(g)) \vee n_{p_2(g)} \in I(p_2(g)) \text{ (according to the IH)} \\
& \Leftrightarrow n_g \in I(p_1(g)) \vee n_g \in I(p_2(g)) \text{ (according to the IH)} \\
& \Leftrightarrow n_g \in I(g) \\
\end{align*}
\textit{Case 2:} Let $\alpha(g) = \cap$. We show $(z \in I(g) \Leftrightarrow n_g \in I(g))$ analogous to Case 1.
\\
\textit{Case 3:} Let $\alpha(g) = \compl$. We show $(z \in I(g) \Leftrightarrow n_g \in I(g))$ similar to Case 1.
\\
\textit{Case 4:} Let $\alpha(g) = +$. The following statement (1) holds:
$$ 2^{|C_g|} > 2^{|C_{p_1(g)}|} + 2^{|C_{p_2(g)}|} $$
\textit{Case 4 a:} Let $n_{p_1(g)} \notin I(p_1(g))$ and $n_{p_2(g)} \notin I(p_2(g))$. According to the induction hypothesis, $I(p_1(g)) \subseteq [0, 2^{|C_{p_1(g)}|}]^m$ and $I(p_2(g)) \subseteq [0, 2^ {|C_{p_2(g)}|}]^m$. Because of (1) we get:
$$I(g) = I(p_1(g)) + I(p_2(g)) \subseteq [0, 2^{|C_{p_1(g)}|} + 2^{|C_{p_2(g)}|}]^m \subset [0, 2^{|C_{g}|}]^m$$
Hence $z \notin I(g)$ and $n_g \notin I(g)$.
\\
\textit{Case 4 b:} Let $n_{p_1(g)} \in I(p_1(g))$ and $n_{p_2(g)} \notin I(p_2(g))$.
For each $y = (y_1, \dots , y_m) \in I(p_2(g))$, $y_i < 2^{|C_{p_2(g)}|} + 1$ and therefore we get (2):
$$ z_i - y_i \geq (2^{|C_g|} + 1) - y_i > (2^{|C_g|} + 1) - (2^{|C_{p_2(g)}|} + 1) \geq 2^{|C_{p_1(g)}|} + 1$$
Therefore we get:
\begin{align*}
& z \in I(g) \\
& \Leftrightarrow \exists y \in \natNum^m: (z-y) \in I(p_1(g)) \wedge y \in I(p_2(g)) \\
& \Leftrightarrow \exists y \in \natNum^m: n_{p_1(g)} \in I(p_1(g)) \wedge y \in I(p_2(g)) \text{ (according to (2) and the IH)}\\
& \Leftrightarrow \exists y \in \natNum^m: (n_g - y) \in I(p_1(g)) \wedge y \in I(p_2(g)) \text{ (according to (2) and the IH)}\\
& \Leftrightarrow n_g \in I(g)
\end{align*}
\\
\textit{Case 4 c:} Let $n_{p_1(g)} \notin I(p_1(g))$ and $n_{p_2(g)} \in I(p_2(g))$. We show $(z \in I(g) \Leftrightarrow n_g \in I(g))$ analogously to Case 4 b. 
\\
\textit{Case 4 d:} Let $n_{p_1(g)} \in I(p_1(g))$ and $n_{p_2(g)} \in I(p_2(g))$. According to (1), the following statement (3) applies:
$$z_i-( 2^{|C_{p_2(g)}|} + 1) \geq (2^{|C_g|} +1)- (2^{|C_{p_2(g)}|}+1) \geq 2^{|C_{p_1(g)}|}+1 $$
Define
\begin{align*}
a \defEquals & ( z_1 &,& \dots &,& z_{i-1} &,& z_i-( 2^{|C_{p_2(g)}|} + 1) &,& z_{i+1} &,& \dots &,& z_m &) \\
b \defEquals & (0   &,& \dots &,& 0       &,& 2^{|C_{p_2(g)}|} + 1      &,& 0       &,& \dots &,& 0   &)  
\end{align*}
Hence $a \in I(p_1(g))$ according to (3) and the induction hypotheses, $b \in I(p_2(g))$ according to the hypotheses and thus $a+b = z \in I(g)$.
Define:
\begin{align*}
c \defEquals & \left( \left( (2^{|C_g|} +1)- (2^{|C_{p_2(g)}|}+1)\right) , \dots, \left( (2^{|C_g|} +1)- (2^{|C_{p_2(g)}|}+1)\right) \right)  \in \natNum^m\\
d \defEquals & (2^{|C_{p_2(g)}|} +1, \dots, 2^{|C_{p_2(g)}|}+1) \in \natNum^m
\end{align*}
Hence $c \in I(p_1(g))$ according to (3) and the induction hypotheses, $d \in I(p_2(g))$ according to the induction hypotheses and thus $c+d = n_g \in I(g)$ .
\\
\textit{Case 5:} Let $\alpha(g) = -$.
For all $y = (y_1, \dots, y_m) \in \natNum^m$ the following statement (4) applies:
$$ z_i + y_i \geq ( 2^{|C_{g}|} + 1) +  y_i \geq ( 2^{|C_{g}|} + 1) \geq 2^{|C_{p_1(g)}|} + 1$$
\begin{align*}
& z \in I(g) \\
& \Leftrightarrow \exists y \in \natNum^m: z+y \in I(p_1(g)) \wedge y \in I(p_2(g)) \\
& \Leftrightarrow \exists y \in \natNum^m: n_{p_1(g)} \in I(p_1(g)) \wedge y \in I(p_2(g)) \text{ (according to (4) and the IH)}\\
& \Leftrightarrow \exists y \in \natNum^m: (n_g + y) \in I(p_1(g)) \wedge y \in I(p_2(g)) \text{ (according to (4) and the IH)}\\
& \Leftrightarrow n_g \in I(g)
\end{align*}
This shows that the auxiliary statement holds for $g$. 

The lemma follows from the auxiliary statement because $|C_g| \leq |C|$ for all $g \in V$.
\end{proof}

\noindent We solve $\MC^*(\cup, \cap, \compl, +, - )$ on an alternating \TM in polynomial time. The approach is similar to the algorithm in \Cref{Algo1}. The modified algorithm can be found in \Cref{Algo2}. 

\begin{figure}
\algrenewcommand\algorithmicrequire{\textbf{Input:}}
\algrenewcommand{\algorithmiccomment}[1]{\State {\color{blue}// #1 }}
\begin{algorithmic}[1]
   \Require $(C,b, m)$
   \State Let $n = 2^{|C|} + 1$, $g = g_C$, $x = b$ and $t = 1$.
   \While{$\alpha(g) \notin \mathbb{N}^m \cup \{ \infty \}$}
	  \If{$x \neq \infty$}      
      	\State Let $(x_1, \dots, x_m) = x$.   
      	\If{$\exists i \in \{1, \dots, m \}: x_i \geq n $}
   			\State Let $x = (n, \dots, n) \in \mathbb{N}^m$. \label{L2_Komplement}
      	\EndIf
      \EndIf      
      \Comment{If $t==1$ test $x \in I(g)$ and if $t==0$ test $x \notin I(g)$.}      
      \If{$\alpha(g) == \cup $} \Comment{Union}
        \If{$\alpha(g) == \cup $ and $t == 0$}
      	  \State Set $g = p_1(g)$ and $g = p_2(g)$ in parallel using a universal state.
        \ElsIf{$\alpha(g) == \cup $ and $t == 1$}
      	  \State Guess $i \in \{ 1, 2 \}$.  
          \State Let $g = p_i(g)$.
        \EndIf
      \ElsIf{$\alpha(g) == \overline{\phantom{I}} $} \Comment{Complement}
        \State Let $g = p(g)$ and $t = 1-t$.      
      \ElsIf{$\alpha(g) == + $} \Comment{Addition}
        \If{$t == 1$}  
          \Comment{Guess $a \in [0,n]^m \cup \{ \infty \}$.}
        		\State Guess $j \in \{0,1\}$. \label{L2_Addition1}
        		\State Let $a =  \infty$.
        		\If{$j == 1$} 
        		  \ForAll{$i \in \{1, \dots, m \}$}
		        \State Guess $a_i \in [0, n] \cap \mathbb{N}$.         
              \EndFor 
              \State Let $a = (a_1, \dots, a_m)$.
        		\EndIf
        		\Comment{Guess $b \in [0,n]^m \cup \{ \infty \}$.}
        		\State Guess $k \in \{0,1\}$.
        		\State Let $b =  \infty$.
        		\If{$k == 1$}
        		  \ForAll{$i \in \{1, \dots, m \}$}
		        \State Guess $b_i \in [0, n] \cap \mathbb{N}$.         
              \EndFor 
              \State Let $b = (b_1, \dots, b_m)$.
        		\EndIf
        		\Comment{Ensure $a + b == x$.}
        		\If{$a + b \neq x$}
        			\State Let $a =  x$ and $b = (0, \dots, 0) \in \mathbb{N}^m$.  
        		\EndIf
        		\State Set $(x,g) = (a, p_1(g))$ and  $(x,g) = (b, p_2(g))$ in parallel using a universal state.

    \algstore{alog2break1}
\end{algorithmic}
\end{figure}

\begin{figure}
\algrenewcommand{\algorithmiccomment}[1]{\State {\color{blue}// #1 }}
\begin{algorithmic}[1]
    \algrestore{alog2break1}
    \Else
    \Comment{$t == 0$, Set $a$ to every value of $[0,n]^m \cup \{ \infty \}$ in parallel.}
        		\State Set $j$ to both values of $\{0,1\}$ in parallel using a universal state. \label{L2_Addition0}
        		\State Let $a =  \infty$.
        		\If{$j == 1$} 
        		  \ForAll{$i \in \{1, \dots, m \}$}
		        \State Set $a_i$ to every value of $[0, n] \cap \mathbb{N}$ in parallel using universal states.         
              \EndFor 
              \State Let $a = (a_1, \dots, a_m)$.
        		\EndIf
        		\Comment{Set $b$ to every value of $[0,n]^m \cup \{ \infty \}$ in parallel.}
        		\State Set $k$ to both values of $\{0,1\}$ in parallel using a universal state.
        		\State Let $b =  \infty$.
        		\If{$k == 1$}
        		  \ForAll{$i \in \{1, \dots, m \}$}
		        \State Set $b_i$ to every value of $[0, n] \cap \mathbb{N}$ in parallel using a universal state.         
              \EndFor 
              \State Let $b = (b_1, \dots, b_m)$.
        		\EndIf
        		\Comment{Ensure $a + b == x$.}
        		\If{$a + b \neq x$}
        			\State Let $a =  x$ and $b = (0, \dots, 0) \in \mathbb{N}^m$.  
        		\EndIf
        		\State Guess $i \in \{1,2\}$.
        		\If{$k == 1$}
        			\State $(x,g) = (a, p_1(g))$.
        		\Else 
        			\State $(x,g) = (b, p_2(g))$.
        		\EndIf         
        	  \EndIf
      \Else \Comment{Subtraction}
        \If{$t == 1$}  
      	  \ForAll{$i \in \{1, \dots, m \}$} \label{L2_Subtraktion1}
		    \State Guess a binary encoded $a_i \in [0, n] \cap \mathbb{N}$.         
          \EndFor
          \State Let $a = (a_1, \dots, a_m)$.
          \State Set $(x,g) = (x + a, p_1(g))$ and  $(x,g) = (a, p_2(g))$ in parallel using an universal state. 
          \Comment{It is $x + a == \infty$, if $x == \infty$}        
        \Else  
          \Comment{$t == 0$} 
          \ForAll{$i \in \{1, \dots, m \}$} \label{L2_Subtraktion0}
		    \State Set $a_i$ to every value of $[0, n] \cap \mathbb{N}$ in parallel using universal states.         
        	  \EndFor 
          \State Let $a = (a_1, \dots, a_m)$.
          \State Guess $i \in \{ 1,2 \}$.
          \If{$i == 1$}
            \State Let $(x,g) = (a + x, p_1(g))$ 
            
    \algstore{alog2break2}
\end{algorithmic}
\end{figure}

\begin{figure}
\algrenewcommand{\algorithmiccomment}[1]{\State {\color{blue}// #1 }}
\begin{algorithmic}[1]
    \algrestore{alog2break2}
    \Comment{It is $x + a == \infty$, if $x == \infty$}      		
          \Else
            \State Let $(x,g) = (a, p_2(g))$.
          \EndIf
        
        \EndIf
      \EndIf
   \EndWhile
   \Comment{g is an input gate.}
    \If{$t == 1$}
      \State Accept, if $x == \alpha(g)$. Otherwise, reject.
   \Else
      \State Reject, if $x == \alpha(g)$. Otherwise, accept.
   \EndIf
\end{algorithmic}
\caption{Algorithm for circuits with union, complement, addition and subtraction over $\mathbb{N}^m \cup \{ \infty \}$}
\label{Algo2}
\end{figure}

\begin{theorem} \label{Satz11}
$\MC^*(\cap, \cup, \compl, +, - ) \in \PSPACE$.
\end{theorem}

\begin{corollary} \label{Folg5} \begin{enumerate}
\item[(i)] $\MC(\{ \cup, \cap , \compl, \times, / \}) \in \PSPACE$.
\item[(ii)] For each $\cO \in \left\lbrace \{ \cup, \cap , \compl, \times \}, \{ \compl, \times \}, \{ \cup, \cap , \times \} \right\rbrace $ the following applies: \\
(a) $\MC(\cO \cup \{ / \})$ is $\PSPACE$-complete. (b) $\MC(\cO)\logEquiv \MC(\cO \cup \{ / \})$.
\end{enumerate}
\end{corollary}
\begin{proof}
(i) We conclude the statement from Theorems \ref{Satz9} (ii) and \ref{Satz11}. \\
(i) For $\cO = \{\compl, \times \}$, Barth et al. showed that $MC(\cO)$ is $\PSPACE$-hard \cite[Lem. 26]{TR17:2} and thus $\PSPACE$-complete. In the other cases, the $\PSPACE$-completeness was already known from McKenzie and Wagner~\cite[Thm. 5.5]{MW07}. Hence $MC(\cO \cup \{ / \})$ is $\PSPACE$-hard in each case. With this, the $\PSPACE$-completeness of $MC(\cO \cup \{ / \})$ follows from (i). Since both $\MC(\cO)$ and $\MC(\cO \cup \{ / \})$ are $\PSPACE$-complete, $\MC(\cO)\logEquiv \MC(\cO \cup \{ / \})$.
\end{proof}

\section{Lower Bounds} \label{Kapitel2}
By a lower bound we mean a complexity class $\mathcal{C}$ such that the considered membership problem is $\mathcal{C}$-hard.

\subsection{Circuits with addition, multiplication and division}
We will show that $\MC(+, \times, /)$ is $\PIT$-hard. 
\begin{definition}
1. $\mathrm{PIT}$ (Polynomial Identity Testing) is the following problem:
We consider the set of $\{ +, \times \}$-circuits $C$ over integers with variables $x_1, \dots, x_n$ as as inputs. We associate the circuit $C$ with the polynomial, that has the variables $x_1, \dots, x_n$. $C$ is in $\mathrm{PIT}$ \ifAndOnlyIf $I(C(x_1, \dots, x_n)) = \{ 0 \}$ for all assignments $x_1, \dots, x_n \in \mathbb {Z}$ is. \\
2. $\PIT$ is the class of problems that are $\logReduction$-reducible to $\mathrm{PIT}$.
\end{definition}
We get the following corollary from \cite[Thm. 45]{TR17:2}:
\begin{corollary} \label{Folg3}
$\MC(\cap,+, \times, /)$ is $\PIT$-hard. 
\end{corollary} 

\noindent We reduce $\MC(\cap,+, \times, /)$ to $\MC(+, \times, /)$ by simulating $\cap$-gates with multiplication and division:
\begin{lemma} \label{Lemma10}
$\MC( \cap, \times, +, /) \logEquiv \MC( \times, +, / )$.
\end{lemma} 
\begin{proof}
We show $\MC(\cap, \times, +, / ) \logReduction \MC(\times, +, / ) $. Note that the sets described by the gates of $\{ \cap, \times, +, / \}$-circuits are at most of cardinality $1$. Let $C$ be such a $\{ \cap, \times, +, / \}$-circuit. We construct a $\{ \times, +, / \}$-circuit $C'$, add one $1$-gate $c_1$ and rebuild each $\cap$-gate $g$ from $C$ as follows: 
We insert new $+$-gates $a_1$ and $a_2$ and new  $/$-gates $d_1$ and $d_2$. We turn $g$ into a $\times$-gate and connect the gates as shown in \Cref{fig:Lemma10}. (Note, that the picture shows multiple $1$-gates to avoid overlapping edges, but we can reuse $c_1$ every time we need a $1$-gate.) All other gates, including their predecessors, are transferred unchanged from $C$ to $C'$. 

We show that $I_C(g) = I_{C'}(g)$ inductively for all gates $g$ from $C$: \\ 
\underline{Base case:} Consider an input gate $g$ of $C$. Then $I_C(g) = \{ \alpha_C(g) \} = \{ \alpha_{C'}(g) \} = I_{C'}(g)$. \\
\underline{Induction step:} We only need to considers $\cap$-gate since all other gates remain unchanged. Let $g$ be an $\cap$-gate from $C$. We assume $I_C(p_1(g)) = I_{C'}(p_1(g))$ and $I_C(p_2(g)) = I_{C'}(p_2(g))$ (induction hypothesis). If $I_C(p_1(g)) = \emptyset$ or $I_C(p_2(g)) = \emptyset$, then $I_{C'}(d_1) = I_{C'}(d_2) = I_{C'}(g) = \emptyset = I_{C}(g)$. Otherwise there are $x_1, x_2 \in \natNum$ with $I_C(p_1(g)) = \{ x_1 \}$ and $I_C(p_2(g)) = \{ x_2 \}$. In this case the following applies:
\begin{align*}
I_{C'}(d_1) & = \left\{
\begin{array}{ll}
\{ \frac{x_1 +1}{x_2 +1} \} & \textrm{if $(x_2 +1)$ divides $(x_1 +1)$} \\
\emptyset & \textrm{otherwise} \\
\end{array}
\right. \\ 
I_{C'}(d_2) & = \left\{
\begin{array}{ll}
\{ 1 \} & \textrm{if } x_1 = x_2 \\
\emptyset & \textrm{otherwise} \\
\end{array}
\right. \\ 
I_{C'}(g) & = \left\{
\begin{array}{ll}
\{ x_1 \} & \textrm{if } x_1 = x_2 \\
\emptyset & \textrm{otherwise} \\
\end{array}
\right\} = I_{C}(g)
\end{align*}
This proves $I_C(g) = I_{C'}(g)$ for all gates $g$ from $C$. In particular this shows $I(C) = I(C')$. The statement follows, because $C'$ can be constructed from $C$ using logarithmic space. 
\qed
\end{proof} 

\begin{figure}[tb]
\begin{center}
\begin{tikzpicture}
  \newcommand\dx{2}
  \newcommand\dy{1.3}
  \newcommand\textoffset{0.5}
  \newcommand\circlesize{1cm}
  
  \node[draw,circle,fill=white,radius=0.3, minimum size=\circlesize] (c1) at (1*\dx, 4.5*\dy) {$1$};
  \node[draw,circle,fill=white,radius=0.3, minimum size=\circlesize] (c2) at (3.5*\dx, 4.5*\dy) {$1$};
  
  \node[draw,circle,fill=white,radius=0.3, minimum size=\circlesize] (c3) at (0.75*\dx, 2.25*\dy) {$1$};
  
  \node[draw,circle,fill=white,radius=0.3, minimum size=\circlesize, label=left:$p_1(g):$] (g1) at (0*\dx, 4.5*\dy) {$\scriptstyle \phantom{g_1}$};
  \node[draw,circle,fill=white,radius=0.3, minimum size=\circlesize, label=left:$p_2(g):$] (g2) at (2.5*\dx, 4.5*\dy) {$\scriptstyle \phantom{g_2}$};
  
  \node[draw,circle,fill=white,radius=0.3, minimum size=\circlesize, label=left:$a_1:$] (a1) at (0.5*\dx, 3.5*\dy) {$\scriptstyle +$};
  \node[draw,circle,fill=white,radius=0.3, minimum size=\circlesize, label=left:$a_2:$] (a2) at (3*\dx, 3.5*\dy) {$\scriptstyle +$};
  
  \node[draw,circle,fill=white,radius=0.3, minimum size=\circlesize, label=left:$d_1:$] (d1) at (1.75*\dx, 2.25*\dy) {$\scriptstyle /$};  
  \node[draw,circle,fill=white,radius=0.3, minimum size=\circlesize, label=left:$d_2:$] (d2) at (1.25*\dx, 1.25*\dy) {$\scriptstyle /$};
  
  \node[draw,circle,fill=white,radius=0.3, minimum size=\circlesize, label=left:$g:$] (g) at (0.75*\dx, 0*\dy) {$\scriptstyle \times$};
  
  \draw[->] (g1) to (a1);
  \draw[->] (c1) to (a1);
  
  \draw[->] (g2) to (a2);
  \draw[->] (c2) to (a2);
  
  \draw[->] (a1) to (d1);
  \draw[->] (a2) to (d1);
  
  \draw[->] (d1) to (d2);
  \draw[->] (c3) to (d2);
  
  \draw[->] (g1)[out=240,in=135] to (g);
  \draw[->] (d2) to (g);
\end{tikzpicture}
\end{center} 
\caption{Sketch of the gates to rebuild one $\cap$-gate $g$ from $C$ in $C'$ in the proof of \Cref{Lemma10}}
\label{fig:Lemma10}
\end{figure}
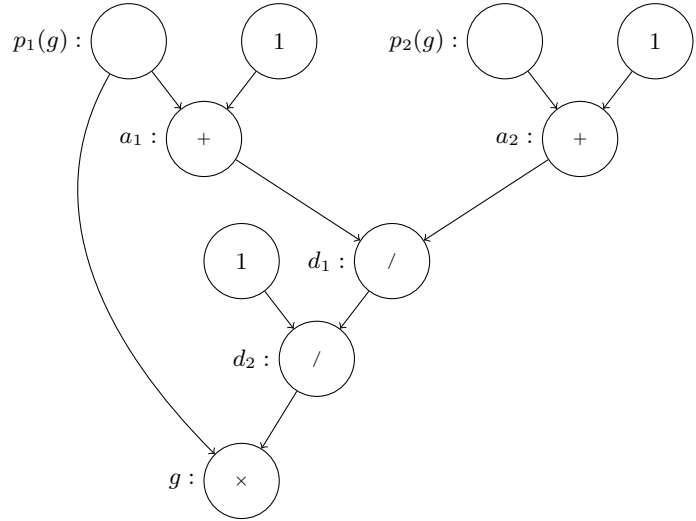

\begin{corollary}\label{Folg7}
$\MC(+, \times, /)$ is $\PIT$-hard. 
\end{corollary} 

\subsection{Circuits with union and division} We will show that $\MC( \cup, /)$ is $\NP$-hard by reducing the \textsc{ExactCover}, which is one of Karp's $\NP$-complete problems \cite{Kar72}, to it. 
\begin{definition}
Let $X$ be a set and $S$ be a set of non-empty subsets of $X$. An \textbf{exact cover} of $X$ is a subset $U \subseteq S$ such that the elements of $U$ are pairwise disjoint and $\bigcup_{A \in U} A = X$. This means that in an exact cover, every element of $X$ is contained in exactly one element of $U$. \ExactCover is the decision problem as to whether such an exact cover exists given $S$ and $X$, that is 
$ \ExactCover \defEquals \{ (S, X) \mid \exists \text{ $U \subseteq S$: $U$ is an exact cover of $X$} \}$.
\end{definition}  

\begin{theorem} \label{Satz5}
$\MC( \cup, /)$ is $\NP$-hard. 
\end{theorem}
\begin{proof}
We will reduce \textsc{ExactCover} to $\MC( \cup, /)$. Let $X = \{a_1, \dots , a_m \}$ be a finite set and $S = \{ A_1, \dots, A_k \} \subseteq \mathcal {P}(X)$ with $\emptyset \not \in S$. Let $p: X \rightarrow \primes$ be the mapping of $a_i$ to the $i$-th prime $p_i$. Because of the Prime Number Theorem \cite[Thm. 6]{Hardy79}, $p_i \in O(\frac{i}{\log_e(i)})$. Hence $p_i$ can be represented in $O(\log(i))$ bits. We can test, whether a positive, natural number $k$ is prime in linear space with respect to the length of the binary representation of $k$, by iterating through all potential factors $l \in [k-1] \setminus \{ 1 \}$ on a different tape of the Turing machine. For each such $l$ we perform the long division (on the binary representation of the numbers) of $k$ divided by $l$ on another tape. Hence the function $p$ can be calculated in \logSpace (with respect to the size of the encoding of $X$ and $S$). For each $A \in \mathcal{P}(X)$ we define $f(A) = \prod_{a \in A} p(a)$. The calculation of such products is possible in \logSpace \cite[Thm. 2.2]{IL95}. Hence $f \in \mathrm{FL}$ and injective. For all $A, B \in S$, $A \subseteq B$ \ifAndOnlyIf $f(A)$ divides $f(B)$ \ifAndOnlyIf $ {f(B)}/{f(A)} \in \natNum$. If $A \subseteq B $ then $ f(B)/f(A) = f(B \setminus A) $. 
We want to construct a $\{ \cup, / \}$-circuit $C$ such that $(S, X) \in \ExactCover$, \ifAndOnlyIf $1 \in I(C)$. For each $A_i \in S$ let $g_{A_i}$ be an input gate with $I(g_{A_i}) = \{ f(A_i) \}$. We define the gates $g_0, \dots g_k$ inductively: For $i=0$, let $g_0$ be an input gate with $I(g_0) = \{ f(X) \}$. We aim for $
I(g_i) = \left\lbrace  f \left( X \setminus \bigcup\nolimits_{j \in I} A_j \right) \mid I \subseteq [i] \text{ and the $A_j$ with $j \in I$ are pairwise disjoint} \right\rbrace 
$ 
for all $i \in [k]$. For $0 \leq i < k$, assume that $g_i$ was already constructed. Let $h_{i+1}$ be a new $/$-gate with first predecessor $g_i$ and second predecessor $g_{A_{i+1}}$. Let $g_{i+1}$ be a new $\cup$-gate with predecessors $g_i$ and $h_{i+1}$. (See \Cref{fig:Satz5}.) Then we get 
{\small \begin{align*}
I(h_{i+1}) & = \left\{ x / f(A_{i+1}) \suchthat x \in I(g_i) \text{ and $f(A_{i+1})$ divides $x$ } \right\} \\
& = \left\{ \tfrac{f \left( X \setminus \bigcup_{j \in I} A_j \right)}{f(A_{i+1})} \suchthat
\begin{array}{ll} I \subseteq [i] \text{and the $A_j$ with $j \in I$ are pairwise } \\
\text{disjoint and $f(A_{i+1})$ divides } f \left( X \setminus \bigcup_{j \in I} A_j \right) \\
\end{array}
\right\} \\
& = \left\{ f \left( X \setminus \bigcup\nolimits_{j \in I} A_j \right) \suchthat
\begin{array}{ll} \{ i +1 \} \subseteq I \subseteq [i +1] \text{ and the $A_j$ } \\
\text{with $j \in I$ are pairwise disjoint} \\
\end{array}
\right\}
\end{align*} }
and $I(g_{i+1})$ is as desired. Let $g_k$ be the output gate of $C$. The following applies:  
\begin{enumerate}
\item If the elements of $U \subseteq S$ are pairwise disjoint, then $f(X \setminus \bigcup_{A \in U} A) \in I(C)$.
\item For all $z \in I(C)$ there is a $U \subseteq S$ such that the elements of $U$ are pairwise disjoint and $z = f(X \setminus \bigcup_{A \in U} A) $.
\item $f(X \setminus \bigcup_{A \in U} A) = 1$ \ifAndOnlyIf $\bigcup_{A \in U} A = X$.
\end{enumerate}
Hence there is an exact cover $U \subseteq S$ of $X$ \ifAndOnlyIf $1 \in I(C)$. $C$ can be constructed using logarithmic space.
Hence $\ExactCover \logReduction \MC( \cup, /)$ and $\MC( \cup, /)$ is $\NP$-hard. \qed
\end{proof}

\begin{figure}
\begin{minipage}[c]{0.35\linewidth}
\begin{center}
\begin{tikzpicture}
	\newcommand\dx{1.2}
	\newcommand\dy{2.25}
	\newcommand\textoffset{0.7}
	\newcommand\circlesize{1cm}
	
	\node[draw,circle,fill=white,radius=0.5, minimum size=\circlesize, label=left:$g_{A_{i+1}}:$] (fi) at (2*\dx, -0*\dy) {\phantom{$/$}};
	\node (fi_label) at (fi) {\tiny\thinmuskip=-3mu\medmuskip=-2mu\thickmuskip=-1mu$f(\hspace{-1.5pt}A_{i+1}\hspace{-1.5pt})$};
	
	\node[draw,circle,fill=white,radius=0.5, minimum size=\circlesize, label=left:$g_i:$] (gi) at (0*\dx, 0*\dy) {\phantom{$/$}};
	
	\node[draw,circle,fill=white,radius=0.5, minimum size=\circlesize, label=left:$h_{i+1}:$] (fii) at (1*\dx, -1*\dy) {$/$};
	
	\node[draw,circle,fill=white,radius=0.5, minimum size=\circlesize, label=left:$g_{i+1}:$] (gii) at (1*\dx, -2*\dy) {$\cup$};
	
	\draw[->] (gi)[out=-120,in=150] to (gii);
	\draw[->] (gi) to (fii);
	\draw[->] (fi) to (fii);
	\draw[->] (fii) to (gii);
\end{tikzpicture}
\end{center}
\caption{Sketch for \Cref{Satz5}}
\label{fig:Satz5}
\end{minipage}
\hfill
\begin{minipage}[c]{0.6\linewidth}
\begin{center}
\begin{tikzpicture}
	\newcommand\dx{1.15}
	\newcommand\dy{1.5}
	\newcommand\textoffset{0.7}
	\newcommand\circlesize{1cm}
	
	\node[circle,fill=white,radius=0.3, minimum size=\circlesize] (c1) at (-1.7*\dx, -0*\dy) {$C^x:$};
	
	\node[draw,circle,fill=white,radius=0.3, minimum size=\circlesize, label=left:$g_1:$,label=right:$2^n$] (c1) at (0*\dx, -0*\dy) { };
	\node[draw,circle,fill=white,radius=0.3, minimum size=\circlesize, label=left:$g_2:$,label=right:$2^m$] (c2) at (2*\dx, -0*\dy) { };
	\node[draw,circle,fill=white,radius=0.3, minimum size=\circlesize, label=left:$f_1:$] (const) at (4*\dx, -0*\dy) {$2$};
	
	\node[draw,circle,fill=white,radius=0.3, minimum size=\circlesize, label=left:$f_2:$,label=right:$2^{m+1}$] (times) at (3*\dx, -1*\dy) {$\times$};
	
	\node[draw,circle,fill=white,radius=0.3, minimum size=\circlesize, label=left:$f_3:$,label=right:$2^n / 2^{m+1}$] (div) at (2*\dx, -2*\dy) {$/$};
	
	\node[draw,circle,fill=white,radius=0.3, minimum size=\circlesize, label=left:$f_4:$] (c) at (2*\dx, -3*\dy) {$/$};
	
	\draw[->] (c1) to (div);
	\draw[->] (c2) to (times);
	\draw[->] (const) to (times);
	\draw[->] (times) to (div);
	
	\draw[->] (div)[out=-45,in=45] to (c);
	\draw[->] (div)[out=-135,in=135] to (c);
\end{tikzpicture}
\end{center}
\caption{The circuit $C^x$ from the proof of Thm. \ref{Satz4}} \label{fig:Satz4}
\end{minipage}
\end{figure}

\subsection{Circuits with complement and division} We can show that $\MC(\compl, /)$ is $\mathrm{P}$-hard by reducing the Circuit Value Problem to it. 
\begin{definition} The \emph{Circuit Value Problem} is defined as
\begin{align*}
\mathrm{CVP} \defEquals \{ (C, x_1, \dots, x_n) \mid & \text{$C$ is a $\{\vee, \wedge, \neg \}$-circuit with $n$ input gates such that } \\
                                                  &  \text{the output is $1$, if the values $x_1, \dots, x_n \in \{ 0, 1 \}$ are } \\
                                                  &  \text{assigned to the input gates} \}. 
\end{align*}
\end{definition}
It is known that $\mathrm{CVP}$ is $\mathrm{P}$-complete \cite[Lem. 4.85]{Vollmer99}.

\begin{theorem} \label{Satz8} $\MC(\compl, /)$ is $\mathrm{P}$-hard. 
\end{theorem}

\begin{proof}
Because of de-Morgan's rule, it is sufficient to show that $\mathrm{CVP}$ can be reduced to $\MC(\compl, /)$ for $\{ \wedge, \neg \}$-circuits. Let $C$ be such a circuit with an output gate $g_C$, input gates $g_1, \dots, g_n$ and inputs $x_1, \dots, x_n \in \{ 0, 1 \}$. We construct a $\{ \compl, / \}$-circuit $C'$ from $C$. We want to represent the values $0$ and $1$ by the empty set and its complement (the set of natural numbers). Add three new gates $f, f_0$ and $f_1$ with $\alpha(f) = 0$, $\alpha(f_0) = / $ and $\alpha(f_1) = \compl$ to $C'$. Let $f$ be the first and second predecessors of $f_0$ and let $f_0$ be the predecessor of $f_1$ (see \Cref{fig:Satz8}). Hence $I(f_0) = \{ 0 \} / \{ 0 \} = \emptyset$ and $I(f_1) = \overline{ \emptyset } = \natNum$. We copy all gates except the input gates from $C$ to $C'$. We replace $\wedge$-gates with $/$-gates and $\neg$-gates with $\compl$-gates. (The order of the predecessors of the $/$-gates may be chosen arbitrarily.) If a gate $g$ in $C$ had an input gate $g_i$ as a predecessor, we add an edge from $f_{x_i}$ to $g$. 
Note how the following equations for sets $\natNum$ and $\emptyset = \overline{\natNum}$ and the truth values $0$ and $1$ correspond: 
{ \small \begin{align*}
\emptyset / \emptyset = \emptyset / \natNum = \natNum / \emptyset & = \emptyset 
& \natNum / \natNum & = \natNum
& \overline{ \emptyset } & = \natNum
& \overline{ \natNum } & = \emptyset \\
0 \wedge 0 = 0 \wedge 1 = 1 \wedge 0 & = 0
& 1 \wedge 1 & = 1
& \neg 0 & = 1
& \neg 1 & = 0
\end{align*} }
We can show by induction that for all gates $g$ from $C'$ with $g \notin \{ f, f_0, f_1 \}$, 
\begin{enumerate}
\item[(i)] $I(g) \in \{ \emptyset, \natNum \}$ and
\item[(ii)] $I(g) = \natNum$ \ifAndOnlyIf the value $1$ is produced at the gate $g$ in $C$ when the inputs are assigned the values $x_1, \dots, x_n$.
\end{enumerate}
Therefore it holds that: 
{ \small $$
(C, x_1, \dots, x_n) \in \mathrm{CVP} 
\; \Leftrightarrow \; I(g_C) = \natNum 
\; \Leftrightarrow \; 1 \in I(g_C) 
\; \Leftrightarrow \; (C', 1) \in \MC(\compl, /)
$$ }
$C'$ can be constructed using logarithmic space. Hence $\mathrm{CVP} \logReduction \MC(\compl, /)$.
\qed
\end{proof}

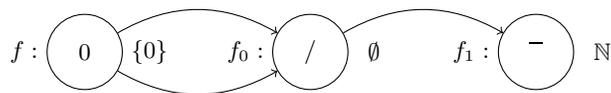
\begin{figure}
\begin{center} 
\begin{tikzpicture}
	\newcommand\dx{3}
	\newcommand\dy{2.5}
	\newcommand\textoffset{0.7}
	
	\node[draw,circle,fill=white,radius=0.5, minimum size=1cm, label=left:$f:$, label=right:$\{0\}$] (f) at (0*\dx, -0*\dy) {$0$};
	
	\node[draw,circle,fill=white,radius=0.5, minimum size=1cm, label=left:$f_0:$, label=right:$\;\,\emptyset$] (f0) at (1*\dx, 0*\dy) {$/$};
	
	\node[draw,circle,fill=white,radius=0.5, minimum size=1cm, label=left:$f_1:$, label=right:$\;\,\natNum$] (f1) at (2*\dx, 0*\dy) {$\overline{\phantom{0}}$};
	
	\draw[->] (f)[out=30,in=150] to (f0);
	\draw[->] (f)[out=-30,in=-150] to (f0);
	\draw[->] (f0)[out=30,in=150] to (f1);
\end{tikzpicture}
\end{center} 
\caption{Sketch for the proof of \Cref{Satz8}}
\label{fig:Satz8}
\end{figure}

\subsection{Circuits with multiplication and division}

\begin{theorem} \label{Satz4} $\MC( \times, /)$ is $\PL$-hard. 
\end{theorem}

\begin{proof}
Let $A \in \PL$. This means that a probabilistic \TM $M$ exits that recognizes $A$ in polynomial time (with runtime $p$) using at most logarithmic space \cite{Jung85}. Hence for all $x$ it holds that $ x \in A$ \ifAndOnlyIf $prob_{M}(x) > \frac{1}{2}$.
Let $M'$ be a probabilistic \TM that simulates $M$ while counting the branching steps on an additional tape. If $M$ with input $x$ stops in a state $q_S$ after less then $p(|x|)$ branching steps, $M'$ continues to make additional branching steps to ensure that the computation tree of $M'(x)$ has exactly $p(|x|)$ branching configurations on each path in the computation tree for the input $x$. $M'$ accepts $x$, if $q_S$ is an accepting state. Hence for each input $x$, $prob_{M'}(x) = prob_{M}(x)$ holds and each computation with the input $x$ ends after exactly $p(|x|)$ branching steps. Like $M$, $M'$ uses at most logarithmic space. 

We will construct a $\{\times, / \}$-circuit $C^x$ such that we can use the membership problem for $C^x$ to check if the input $x$ is element of $A$. Let $f_1, \dots, f_4$ and $g_1$ and $g_2$ be gates of $C^x$. Let $\alpha(f_1) = 2, \alpha(f_2) = \times$ and $\alpha(f_3) = \alpha(f_4) = /$. Let $g_2$ and $f_1$ be the predecessors of $f_2$, $g_1$ be the first and $f_2$ the second predecessors of $f_3$, and $f_3$ be both predecessors of $f_4$. Let $f_4$ be the output gate of $C^x$. A sketch of $C^x$ can be found in \Cref{fig:Satz4}. The configuration graph of $M'$ with input $x$ can be constructed using logarithmic space. While we construct the configuration graph we simultaneously translate it into two $\{\times\}$-circuits $C_1^x$ and $C_2^x$ such that $g_1$ is the output gate of $C_1^x$ and $g_2$ is the output gate of $C_2^x$. The gates $g_1$ and $g_2$ will represent the initial configuration of $M'$ with input $x$. 
As we construct the configuration graph we replace each configuration in which $M'$ does not stop with a $\times$-gate in $C_1^x$ and another $\times$-gate in $C_2^x$. For each configuration that is followed by a deterministic step we introduce a $1$-gate as one of its predecessors. Therefore each $\times$-gate in  $C_1^x$ and $C_2^x$ has two predecessors. We replace each configuration in which $M'$ stops by an input gate. In $C_1^x$ each accepting configuration is represented by a $2$-gate and each rejecting configuration is represented by a $1$-gate. In $C_2^x$ we label the input gates the other way around. 
Let $n$ be the number of accepting paths and $m$ be the number of rejecting paths of $M'$ with input $x$. Hence $I(C_1^x) = 2^n$ and $I(C_2^x) = 2^m$. Furthermore for each input $x$, $prob_{M'}(x) = \frac{n}{2^{p(|x|)}} $ and $M'$ rejects $x$ with the probability $\frac{m}{2^{p(|x|)}}$.
The following applies:
{ \small $$ x \in A \; \Leftrightarrow \; prob_{M'}(x) > \frac{1}{2} \; \Leftrightarrow \; n > m \;\Leftrightarrow \; 2^n \geq 2^{m +1} \; \Leftrightarrow \; 2^n / 2^{m +1} \in \natNum \; \Leftrightarrow \; 1 \in I(C^x)$$ }
Since $C^x$ can be constructed in \logSpace, we get $A \logReduction \MC( \times, /)$. \qed
\end{proof}

\subsection{Division only circuits}
We can show, that $\MC(/)$ is $\NL$-hard by reducing the Graph Accessibility Problem to the complement of $\MC(/)$. 
\begin{definition} The \emph{Graph Accessibility Problem} is defined as
\begin{align*}
\mathrm{GAP} = \{ (G, u, v) \mid & \text{$G$ is a directed graph that has a path form $u$ to $v$} \}.
\end{align*}
\end{definition}
It is known that $\mathrm{GAP}$ is $\NL$-complete \cite[Thm. 8.20]{Sipser97}. 

\begin{theorem} \label{Satz2} $\MC(/)$ is $\NL$-hard.
\end{theorem}
\begin{proof} 
Let $G = (V, E)$ be a directed acyclic graph with $s,t \in V$. We can modify $G$ in \logSpace such that $s$ is a source, $t$ is a sink, and every node $v \in V$ has indegree 0 or 2. We also assume $s \neq t$. Note that vertices with indegree 0 (except $s$) are never part of any $s$-$t$-path (that is, a path from $s$ to $t$). Still, we keep these vertices to ensure that each other gate has indegree 2.
For every node $v$ with $v \neq s$ and indegree 0 let $\alpha(v) = 1$ and for all nodes $w$ with indegree 2 let $\alpha(w) = /$. Also, let $\alpha(s) = 0$.
Then $C = (G, E, t, \alpha)$ is a $\{ / \}$-circuit. The order of the predecessors of the nodes with indegree 2 can be chosen arbitrarily. Inductively one can show for all $v \in V$, $I(v) \in \{ \{0 \}, \{1 \}, \emptyset \}$ and that there is $s$-$v$-path in $G$, \ifAndOnlyIf $I(v) \in \{ \emptyset, \{ 0 \} \} $. This shows that $(G, s, t) \in \mathrm{GAP}$ \ifAndOnlyIf there is an $s$-$t$-path in $G$ \ifAndOnlyIf $I(t) \in \{ \emptyset, \{ 0 \} \}$ \ifAndOnlyIf $(C,1) \in \overline{\MC(/)}$. The problem $\MC(/)$ is $\NL$-hard, because $\mathrm{GAP}$ is $\NL$-complete and $\NL = \mathrm{coNL}$. \qed
\end{proof}

\section{Conclusion and open questions} \label{Kapitel_Overview}
There are both cases in which $\MC( \cO) \logEquiv \MC( \cO \cup \{ / \}) $ and cases in which--if common assumptions hold--this is not the case. In this sense, the following lower bounds are interesting compared to the results by McKenzie and Wagner~\cite{MW07}:
\begin{itemize}
\item $\MC(\cup)$ is $\NL$-complete. $\MC(\cup, /)$ is $\NP$-hard.
\item $\MC(\cup,\cap)$ and $\MC(\cup,\cap,\compl)$ are $\mathrm{P}$-complete. $\MC(\cup,\cap, /)$ and $\MC(\cup,\cap,\compl,/)$ are $\NP$-hard.
\item $\MC(\times)$ is $\NL$-complete. $\MC(\times,/)$ is $\PL$-hard. 
\item $\MC(+,\times)$ is $\mathrm{P}$-hard. $\MC(+,\times, /)$ is $\PIT$-hard. 
\end{itemize}
The question of whether $\MC(\cup, \cap, \compl, +, \times, /)$ and $\MC(\compl, +, \times, /)$ are decidable remains open. 
Several bounds are not tight and closing these gaps remains open. Membership problems for formulas with division have not been studied so far. It would be interesting to also consider subtraction and other forms of division and compare them to the division without remainder and without rounding. I am also looking at emptiness and equivalence problems with division.

\begin{credits}
\subsubsection{\ackname} I'm grateful for the encouraging discussions with my former supervisor Prof. Dr. Christian Glaßer. I would like to thank my friend Nils Morawietz for handy tips on writing a paper and my partner Andreas Sacher for the \LaTeX\ support.

\end{credits}

\newpage
\bibliographystyle{splncs04}
\bibliography{main}

\end{document}